\documentclass[12pt,a4paper]{article}
\title{Finite size emptiness formation probability of the  XXZ spin chain at
  $\Delta=-\frac{1}{2}$} 
\date{\today}
\usepackage{bm,cite,color,verbatim,marvosym,multicol,multirow} 
\usepackage[english]{babel}
\usepackage[dvips]{graphicx}
\usepackage{latexsym,amsfonts,amssymb}
\usepackage{amsmath}
\usepackage{amsthm}
\usepackage{eufrak}
\usepackage{epsf}
\usepackage{array}
\usepackage{arydshln}

\usepackage{xypic}
\usepackage{amscd}
\usepackage[all]{xy}
\usepackage{textcomp} 
\usepackage{makeidx} 
\usepackage{mathrsfs}

\def\eq{\begin{equation}}
\def\en{\end{equation}}

\def\>>{\rangle}
\def\tr{{\rm tr}}

\def\bC{\mathbb C}

\def\cH{\mathcal H}
\def\cG{\mathcal G}

\def\cS{\mathcal S}

\def\cM{\mathcal M}
\def\cD{\mathcal D}
\def\cE{\mathcal E}
\def\cP{\mathcal P}

\newtheorem{proposition}{\emph {Proposition}}
\newtheorem{corollary}{\emph {Corollary}}

\begin{document}	

%%%%%%%%
%%%%Text
%%%%%%%%

%%%%%%%%
%%%%Text
%%%%%%%%
\thispagestyle{empty}

\begin{titlepage}

\maketitle

\vspace{-0.5 truecm}

\centerline{\large Luigi
  Cantini\footnote[1]{LPTM, Universit\'e de Cergy-Pontoise (CNRS UMR
    8089), Site de Saint-Martin,
2 avenue Adolphe Chauvin, F-95302 Cergy-Pontoise Cedex, France.  {\small \tt
      <luigi.cantini@u-cergy.fr>}}
    }

\vspace{.3cm}

\vspace{1.0 cm}

\begin{abstract}
\noindent
In this paper we compute the Emptiness Formation Probability of a
(twisted-) periodic XXZ spin chain of finite length at
$\Delta=-\frac{1}{2}$, thus proving the formulas conjectured by
Razumov and Stroganov \cite{raz-strog1, raz-strog2}. The result is
obtained by exploiting the fact that the ground state of the
inhomogeneous XXZ spin chain at $\Delta=-\frac{1}{2}$ satisfies a set
of qKZ equations associated to $\displaystyle{U_q(\hat{sl_2})}$.

\end{abstract}

\end{titlepage}

\section{Introduction}

The computation of correlation functions in quantum integrable
systems, is in general a quite hard task. One paradigmatic example is
the spin-$\frac{1}{2}$ XXZ spin chain \cite{Baxter}. Even in the study
of the free
fermion case quite interesting mathematical structures have
appeared \cite{mccoy}. 
Starting from the seminal papers of Razumov and Stroganov
\cite{raz-strog1,raz-strog2}, we know that the anti-ferromagnetic
ground state of the XXZ spin chain at the value  $\Delta=-\frac{1}{2}$
of the anisotropy parameter (or equivalently $q=e^{2\pi i/3}$)
presents 
a remarkable combinatorial structure. 
The spin chain hamiltonians with an odd number of spins $N=2n+1$ and
periodic boundary conditions or even number of spins $N=2n$ and
twisted-periodic boundary conditions are 
related through a change of basis to the Markov matrix of a stochastic
loop model \cite{BdGN}. 
In the loop basis, the ``ground state'' is actually the steady state
probability of the stochastic loop model. 
The most astonishing discovery was made by
Razumov and Stroganov  \cite{raz-strogO(1)_1}, who observed that, once
properly normalized, the  components of the steady state are integer
numbers enumerating Fully Packed Loop configurations on a square grid.
This conjecture has been eventually proven in \cite{cantini-sportiello}. 

Despite being the XXZ spin chain at $\Delta=-\frac{1}{2}$ a fully
interacting system, several of its correlation functions have simple 
exact formulate even at finite size.
This is the case of the Emptiness Formation Probability
(in short EFP), which is the probability that $k$ consecutive spins are in the
up direction  in a chain of length $N$.  
In their original papers \cite{raz-strog1,raz-strog2}, Razumov and
Stroganov have conjectured exact factorized formulas for the
EFP in terms of products of factorials. The aim of the present paper
is to prove these conjectures.  

\vskip .5cm

The XXZ spin chain with odd size has two ground states $\Psi^+_{2n+1}$
and $\Psi^-_{2n+1}$, related by a spin flip on each site. 
Razumov and Stroganov have conjectured \cite{raz-strog1}
that $E^{-}_{2n+1}(k)$, the EFP  of a k-string of spins up in the state
$\Psi^-_{2n+1}$,  satisfies  
\begin{equation}\label{recurE--}
\frac{E_{2n+1}^{-}(k-1)}{E^{-}_{2n+1}(k)}=
\frac{(2k-2)!(2k-1)!(2n+k)!(n-k)!}{(k-1)!(3k-2)!(2n-k+1)!(n+k-1)!}.  
\end{equation}
Strangely enough, Razumov and Stroganov didn't provide the analogous
formula for the state $\Psi^+_{2n+1}$, which reads
\eq\label{EFP++}
\frac{E^{+}_{2n+1}(k-1)}{E^{+}_{2n+1}(k)}=
\frac{(2k-2)!(2k-1)!(2n+k)!(n-k+1)!
}{(k-1)!(3k-2)!(2n-k+1)!(n+k)!}
\en
In particular, the probability of having a string of spins-up of
length $n$ (or $n+1$) in a chain of length $2n+1$ is equal to the inverse
of $A_{HT}(2n+1)$, the number of  of Half Turn 
Symmetric Alternating Sign Matrices of size $2n+1$,  
\eq\label{sum-conj1}
E_{2n+1}^{-}(n)= E_{2n+1}^{+}(n+1)= A_{HT}(2n+1)^{-1} =
\prod_{j=1}^n\frac{(2j-1)!^2(2j)!^2}{(j-1)!j!(3j-1)!(3j)!} 
\en

In the case of a spin chain with even length and twisted
boundary conditions, the ground state is unique. Razumov and Stroganov
have conjectured \cite{raz-strog2} that $E^{e}_{2n}(k)$, the EFP of a
k-string of spins up satisfies
\eq\label{EFPee*}
\frac{E^{e}_{2n}(k-1)}{E^{e}_{2n}(k)}=
\frac{(2k-2)!(2k-1)!(2n+k-1)!(n-k)!}
     {(k-1)! (3k-2)!(2n-k)!(n+k-1)!} 
\en
The previous equation implies that in the case $k=n$ 
\eq\label{sum-conj2}
E^{e}_{2n}(k)=A_{HT}(2n)^{-1}.
\en
Unlike the ground states of the odd size chains, whose components can
be chosen to be real, the even size ground state has complex valued
components, 
therefore we can consider also  $E^{\tilde
  e}_{2n}(k)$, a sort of ``pseudo'' EFP obtained by sandwiching the
ground state with itself 
(and not with its complex conjugate). The ratio of ``pseudo'' EFPs
corresponding to the same size of the spin chain has a factorized form
given by   
\eq\label{EFPee}
\frac{E^{\tilde e}_{2n}(k-1)}{E^{\tilde e}_{2n}(k)}=
-q\frac{(2k-2)!(2k-1)!(2n+k-1)!(n-k)!
}{(k-1)!(3k-3)!(3k-1)(2n-k)!(n+k-1)!}.
\en
It turns out that, apart for the factor $-q$, the ratio in
eq.(\ref{EFPee}) can be written  as a ratio of enumerations of
$k$-Punctured Cyclically Symmetric Self-Complementary Plane Partitions
(PCSSCPP) of size $2n$, i.e. rhombus tilings of a regular hexagon of 
side length $2n$, which are symmetric under a $\pi/3$ rotation around the
center of the hexagon and with a star shaped frozen region of size
$k$,  as exemplified in Figure \ref{figura1}.
\begin{figure}
\begin{center}
  \includegraphics[scale=.6]{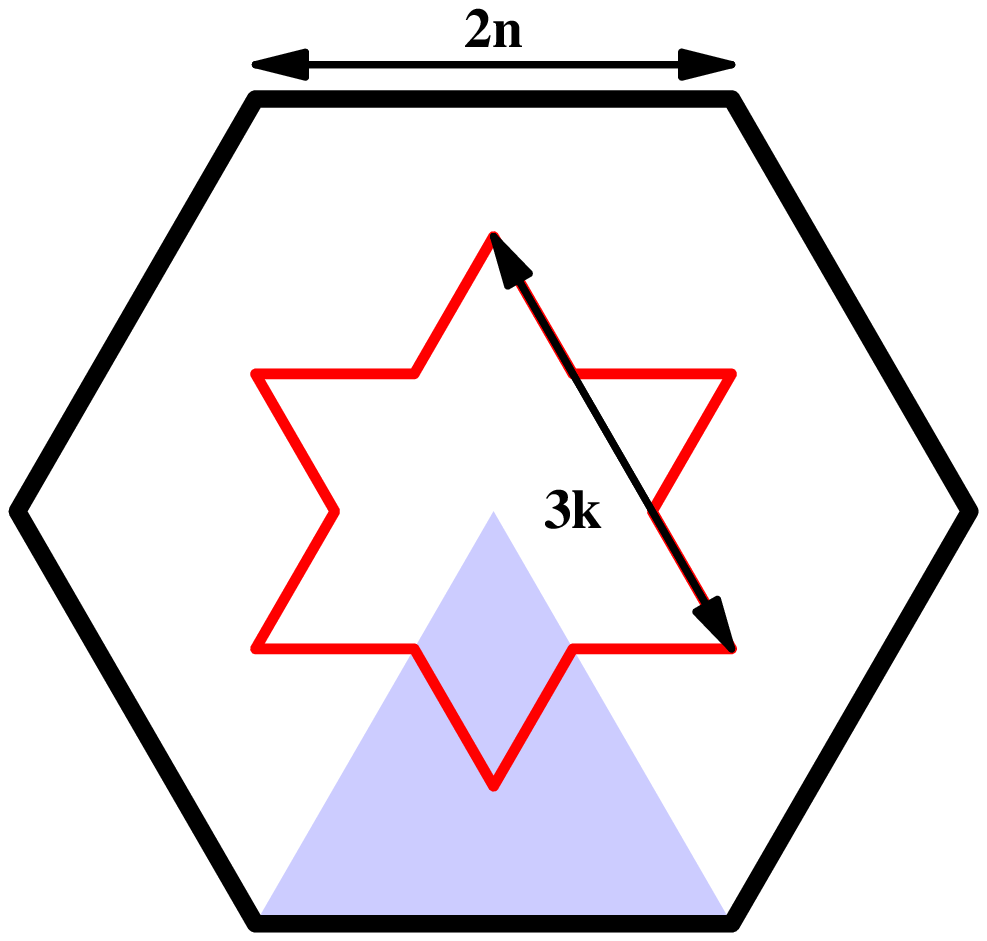}~~~~~~\includegraphics[scale=.6]{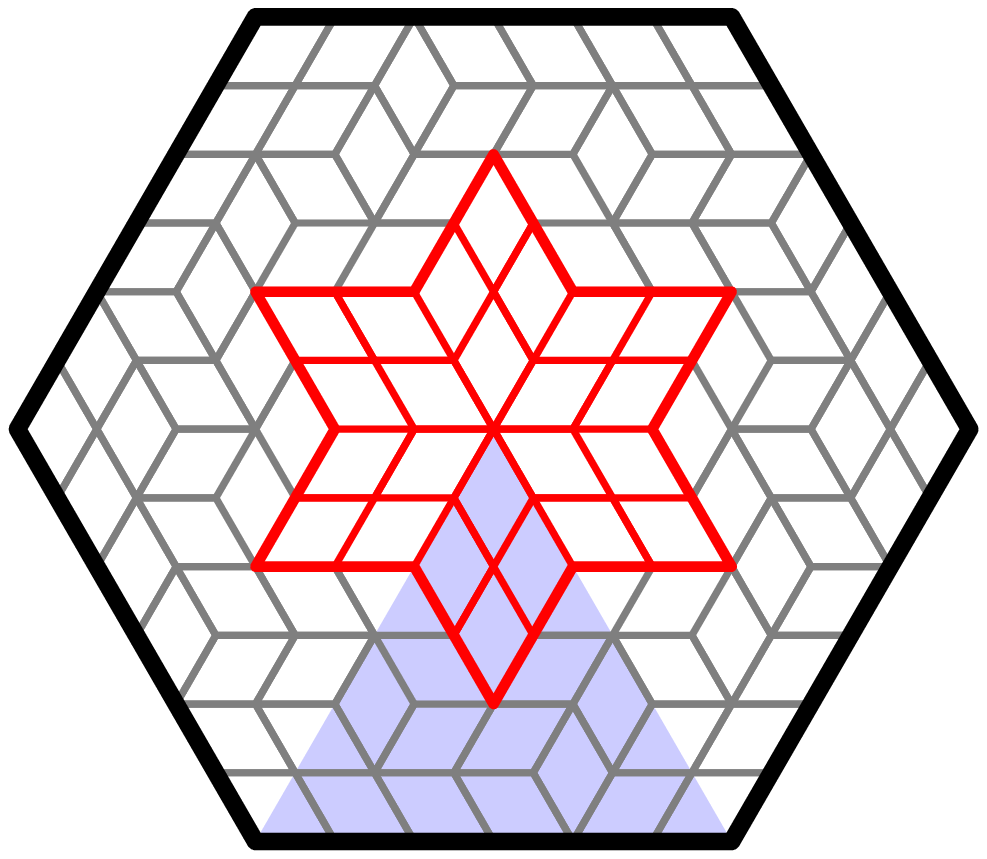}
\end{center}
\caption{Domain tiled by $k$-Punctured Cyclically Symmetric
  Self-Complementary Plane Partitions, with an example of tiling for 
  $n=3$ and $k=2$. The shadowed region indicates a fundamental domain.}\label{figura1}
\end{figure}
Calling these enumerations $CSSCPP(2n,k)$ we have  
\eq
\frac{E^{\tilde e}_{2n}(k-1)}{E^{\tilde e}_{2n}(k)} =-q \frac{CSSCPP(2n,k-1)}{CSSCPP(2n,k)}.
\en
For $k=n$ one obtains 
\eq\label{sum-conj3}
E^{\tilde e}_{2n}(n) = (-q)^{n}A_n^{-2}= (-q)^{n}CSSCPP(2n)^{-1}
\en
where $A_n$ is the number of Alternating Sign Matrices of size $n$
\eq
A_n=\prod_{j=0}^{n-1} \frac{(3j+1)!}{(n+j)!},
\en
and $CSSCPP(2n)$ is the number of  Cyclically Symmetric
Self-Complementary Plane Partitions in an hexagon of size $2n$. 
The enumerations $CSSCPP(2n,k)$ are easily computed by applying a result
of Ciucu \cite{ciucu} concerning enumerations of dimer coverings of
planar graphs with reflection symmetry. This  will be explained
briefly in Appendix \ref{plane-part}.  
 
\vskip .5cm 

Some partial results concerning the (pseudo)-norm, and the EFP have
been obtained in the literature. In \cite{pdf-pzj-jbz}, by
cleverly exploiting  
the relation between a natural degenerate scalar product in the loop basis and the
usual scalar product in the spin basis, Di Francesco and collaborators
have proven eq.(\ref{sum-conj1}) and eq.(\ref{sum-conj3}).

In the large $N$ limit $N\rightarrow \infty$, the EFP have been
studied Maillet and collaborators \cite{maillet-emptiness}. Of course
the conjectures (\ref{recurE--}, \ref{EFP++}, \ref{EFPee*})    
must coincide and must give for the thermodynamic limit of the EFP
\eq
\lim_{N\rightarrow \infty} E_N(k) = \left(\frac{\sqrt{3}}{2}
\right)^{3k^2}\prod_{j=1}^k
\frac{\Gamma(j-1/3)\Gamma(j+1/3)}{\Gamma(j-1/2)\Gamma(j+1/2)}. 
\en
This formula has been proven in \cite{maillet-emptiness} by 
specializing to $\Delta=-\frac{1}{2}$ a
multiple integral representation for the correlation functions, which
is valid for generic values of the anisotropy parameter $\Delta$.

\vskip .5cm

The most effective technique 
which has allowed to compute (partial) sum of 
components in the loop or in the spin basis has been pushed forward by
Di Francesco and Zinn-Justin \cite{pdf-pzj-1} in the context of the 
periodic loop model with an even number of sites. 
They introduced spectral parameters in the model 
in such a way to preserve its integrability, the original model
being recovered once the spectral parameters are set to $1$. In this
way the components of the ground state become homogeneous polynomials
in the spectral parameters, which satisfy certain relation under
exchange or specialization of the spectral parameters. 
As noticed first by Pasquier \cite{pasquier} and largely developed by
Di Francesco and Zinn-Justin \cite{pdf-pzj-qKZ1} the exchange
relations satisfied by the ground state are a special case ($q=e^{2\pi
  i/3}$) of the very much studied quantum 
Knizhnik--Zamolodchikov equations (qKZ)
\cite{Frenkel-Reshetikhin}. 
In \cite{r-s-pzj} the authors have applied this idea to the XXZ spin
chain with spectral parameters and have shown that the properly normalized
ground state of the spin chain with periodic or twisted periodic
boundary conditions satisfies a special case of the $U_q(sl_2)$ qKZ
equations at level $1$. Here we employ this property in order to compute the
Emptiness Formation Probability. Our main idea is to consider a
generalization of the EFP with spectral parameters EFP, which is
constructed from the solution of the $U_q(sl_2)$ qKZ equations for
generic $q$ (see
eqs.(\ref{def-inhom-EFP},\ref{def-inhom-pseudo}). This 
``inhomogeneous'' EFP has certain symmetry and 
recursion properties that completely fix it  
in the same spirit as the recursion relations of the 6-vertex model
with Domain Wall Boundary Conditions completely determine its
partition function.  This will allow us to present an explicit
determinantal formula for the inhomogeneous EFP valid at $q=e^{2\pi
  i/3}$, and upon specialization of the spectral parameters will allow
to obtain the formulas (\ref{recurE--}, \ref{EFP++}, \ref{EFPee*},
\ref{EFPee}).   

\vskip .5cm

The idea to use the solution of the qKZ equation to compute the
inhomogenous version of a correlation function can be in principle
adapted to several other 
models like: XXZ spin chain with different boundary
conditions, fused XXZ spin chain, $U_q(sl_n)$ spin chain or even XYZ
spin chain etc. \cite{boundary,fused,SU(N),XYZ}. Indeed, in all these cases,
by properly tuning the parameters (generalizing the relation
$q=e^{2\pi i/3}$), one has a so called ``combinatorial point'' at which
the ground state energy per site doesn't get finite size
corrections. By reasonings similar to the one in \cite{r-s-pzj}, one
can argue that the ground state \emph{with   spectral parameters}
satisfies a qKZ equation (or, in the case of the XYZ spin chain, an
elliptic version of it). 

Whether this idea could lead to
other exact finite size formulae for some correlation function is an
open question that in our opinion deserves further investigation.

\vskip .5cm

The paper is organized as follows. In Section \ref{conv}, after having
recalled some basic facts about the XXZ spin chain, following
\cite{r-s-pzj} we present the exchange equations satisfied by the
ground state at $\Delta=-\frac{1}{2}$, then in Section
\ref{rec-subsect} we derive the recursion relations satisfied by the
solutions of the $U_q(sl_2)$ qKZ equations at level $1$. In
Section \ref{inhom-section} we define the inhomogeneous version of the
(pseudo) EFP, constructed using the solutions of the qKZ equations. We
derive first its symmetries and then in Section 
\ref{rec-section-q-gen} we derive the recursion relations which
completely determine it. 
In section \ref{combinatorial-pol} we will restrict to $q=e^{2\pi
  i/3}$ and, by showing that certain determinantal expressions satisfy the same
recursion relations as the inhomogeneous EFP, we produce a
representation of this inhomogeneous EFP whose homogeneous
specialization  is considered in Section \ref{hom-sect}, where we prove
the main conjectures. 
In Appendix \ref{fact-det} we compute the
determinants of a family of matrices which are relevant for the
computation of the homogeneous specialization considered in Section
\ref{hom-sect}. In Appendix \ref{plane-part} we compute the lozenge
tilings enumerations $CSSCPP(2n,k)$.

\section{XXZ spin chain at $\Delta=-\frac{1}{2}$ and qKZ equations}\label{conv}

The hamiltonian of the XXZ spin chain acts on a vector space
$\cH_N=(\bC^2)^{\otimes N}$ that
consists of $N$ copies of $\bC^2$ each one labeled by an index
$i$. The hamiltonian is written in terms of 
operators $\sigma^\alpha_i$ which are Pauli matrices acting 
locally on the $i$-th component $\bC_i^2$ 
\eq\label{xxz-ham}
H_N(\Delta) = -\frac{1}{2}\sum_{i=1}^N
\sigma^x_i\sigma^x_{i+1}+\sigma^y_i\sigma^y_{i+1}+\Delta
\sigma^z_i\sigma^z_{i+1}.
\en
It is convenient to parametrize the anisotropy parameter
as $\Delta=\frac{q+q^{-1}}{2}$. 
The model is completely specified once the boundary
conditions are provided. Here we will consider:
\begin{itemize}
\item periodic boundary conditions
for odd values of the length of the spin chain, $N=2n+1$, 
i.e. $\sigma_{N+1}^\alpha = \sigma_{1}^\alpha$ 
\item twisted periodic boundary conditions
for even values of the length of the spin chain, $N=2n$, 
i.e. $\sigma_{N+1}^z = \sigma_{1}^z$, while 
$\sigma^\pm_{N+1}=e^{\pm i \frac{2\pi}{3}}\sigma^\pm_{N+1}$, where $\sigma^\pm=
\sigma^x+\pm i \sigma ^y$.
\end{itemize}
It is well know \cite{Baxter} that the Hamiltonian (\ref{xxz-ham}), for
generic values of the parameter $\Delta$ and of the twisting, is
the logarithmic derivative of an integrable transfer matrix. In order
to define the transfer matrix we need the $R$-matrix and the
twist matrix. 
In the present context the $R$-matrix is
an operator depending on a spectral parameter $z$, which acts on a tensor
product $\bC^2_i\otimes \bC^2_j$. Introducing the basis of $\bC_i^2$
$$
e^{\uparrow}_i=\left(\begin{array}{c}1\\0 \end{array} \right),~~~e^{\downarrow}_i=
\left(\begin{array}{c}0\\1 \end{array} \right) .
$$
we write $R_{i,j}(z)$ in the basis 
$\{e^{\uparrow}_i\otimes
e^{\uparrow}_j,e^{\uparrow}_i\otimes
e^{\downarrow}_j,e^{\downarrow}_i\otimes
e^{\uparrow}_j,e^{\downarrow}_i\otimes e^{\downarrow}_j \} $ of $\bC_i^2
\otimes \bC^2_j$ as 
\eq\label{R-matr}
R_{i,j}(z)= \left( 
\begin{array}{cccc}
a(z) & 0 &0 &0 \\
0 & b(z) &c_1(z) &0 \\
0 & c_2(z) & b(z) &0 \\
0 & 0 &0 &a(z) 
\end{array}
\right)
\en 
with
\eq\label{coef-Rmatr}
a(z)=\frac{qz -q^{-1}}{q-q^{-1}z},~~~~b(x)=\frac{z -1}{q-q^{-1}z},
~~~~c_1(z)=\frac{(q -q^{-1})z}{q-q^{-1}z}, ~~~~c_2(z)=\frac{(q
  -q^{-1})}{q-q^{-1}z}  . 
\en
The twist matrix $\Omega(\phi)$ acts on a single $\bC_i^2$ and in the
basis $(e^{\uparrow}_i,e^{\downarrow}_i$) it reads
\eq
\Omega(\phi) = \left(
\begin{array}{cc}
e^{i\phi} & 0\\
0 & e^{-i\phi}
\end{array}
\right).
\en
Using both the twist and the $R$-matrix we construct the family of
transfer matrices 
\eq\label{mono}
T_N(y|{\bf z}_{\{1,\dots,
  N\}},\phi) = \tr_0\left[R_{0,1}(y/z_1)R_{0,2}(y/z_2)\dots
  R_{0,N}(y/z_N) \Omega_0(\phi) \right]
\en
depending on $N$ ``vertical'' spectral parameters ${\bf z}_{\{1,\dots,
  N\}}$\footnote{Our convention for a ordered string of variables labeled by
  an index is to use a bold character and a label for
  the ordered set of indices of the variables: ${\bf x}_{\{a_1,\dots,
  a_N\}}= \{x_{a_1},\dots,x_{a_N} \}$. Often, when clear from the context, we
  will omit the label ${\{a_1,\dots,a_N\}}$ and write ${\bf x}$ for ${\bf x}_{\{a_1,\dots,
  a_N\}}$.} and a single 
``horizontal'' spectral parameter $y$. Thanks to the commutation
relation
\eq
[R_{i,j}(x),\Omega_i(\phi)\otimes \Omega_j(\phi)]=0
\en
and the Yang-Baxter equation
\eq\label{YBE}
R_{i,j}(z_i/z_j)R_{i,k}(z_i/z_k)R_{j,k}(z_j/z_k)=
R_{j,k}(z_j/z_k)R_{i,k}(z_i/z_k)R_{i,j}(z_i/z_j) 
\en
the transfer matrices for different values of $y$ commute
\eq
[T_N(y'|{\bf z}_{\{1,\dots,N\}},\phi),T_N(y''|{\bf z}_{\{1,\dots,N\}},\phi)]=0.
\en
The hamiltonian of the XXZ spin chain is given by
\eq
\frac{1}{T_N(1|{\bf 1},\phi)}\frac{dT_N(y|{\bf 1},\phi)}{dy}{\Bigg |
}_{y=1} = -\frac{1}{q-q^{-1}} \left( H_N(\Delta) -\frac{3N}{2}\Delta
\right).
\en
At $\Delta=-\frac{1}{2}$ and for generic values of the
vertical spectral parameters, both in the
odd size case with  periodic boundary conditions and in the even size
case with twisted boundary conditions,  the transfer matrix
has an eigenvalue equal to $1$. 
\begin{itemize}
\item 
  When $N=2n+1$, the eigenspace with eigenvalue $1$ is two-fold
  degenerate, $\Psi^\pm_{2n+1}({\bf z})$, corresponding to two the
  values $\pm \frac{1}{2}$ of the total spin
  $S^z=\frac{1}{2}\sum_{i=1}^{N}\sigma^z_i$
\eq
S^z \Psi^\pm_{2n+1}({\bf z})=\pm \Psi^\pm_{2n+1}({\bf z}) 
\en
The two eigenstates are related by a flipping of all the spins
\eq
\Psi^+_{2n+1}({\bf z})  = \prod_{i=1}^N \sigma^x_i \cdot\Psi^-_{2n+1}({\bf z}).
\en

\item 
When $N=2n$ there is a single vector $\Psi^e_{2n}({\bf z})$ with
eigenvalue $1$. It is in the zero sector of the total spin,
\eq
S^z \Psi^e_{2n}({\bf z})=0,
\en
These eigenstates reduce to the
anti-ferromagnetic ground state(s) of the XXZ spin chain when the spectral
parameters are specialized at $z_i=1$. 
\end{itemize}

\subsection{Exchange relations at $\Delta = -\frac{1}{2}$}

A crucial observation made in \cite{r-s-pzj} was that, for an appropriate
choice of the normalization of $\Psi^\mu_N({\bf z})$, the eigenvector equation
\eq
T_N(y|{\bf 1},\phi)\Psi^\mu_N({\bf z}) = \Psi^\mu_N({\bf z})
\en
is equivalent to a set of exchange relations.
Define the exchange operator $P_{i,j}
(e^\mu_i\otimes e^\nu_j) =e^\nu_i\otimes e^\mu_j $, the left rotation
operator $\sigma(v_1\otimes v_2\otimes\cdots \otimes v_{N-1}\otimes
v_N)=v_2\otimes\cdots \otimes v_{N-1}\otimes
v_N\otimes v_1$ and let $\check R_{i,i+1}(z) = P_{i,i+1}R_{i,i+1}(z)$, then
 $\Psi^\mu_N({\bf z})$, as a polynomial  of minimal
degree in the spectral parameters $z_i$, is determined up to a
constant factor, by the following set of equations \cite{r-s-pzj}
\begin{align}\label{qKZ1}
\check R_{i,i+1}(z_{i+1}/z_i)\Psi^\mu_N(z_1,\dots,z_i,z_{i+1},\dots,z_N) &=
\Psi^\mu_N(z_1,\dots,z_{i+1},z_{i},\dots,z_N) \\[5pt]\label{qKZ2}
\sigma
\Psi^\mu_N(z_1,z_2,\dots,z_{N-1},z_N)&=D\Psi^\mu_N(z_2,\dots,z_{N-1},z_N,s
z_1). 
\end{align}
with $D=s=1$. These equations can be seen as the special case
$q=e^{2\pi i/3}$ of the level $1$ qKZ equations
\cite{Frenkel-Reshetikhin}, which
corresponds to generic $q$, $s=q^6$ and
$D=q^{3N}q^{3(s^z_N+1)/2}$. The solution of the level $1$ 
qKZ equations can be normalized in such a way that they become
polynomials in the variables $z_i$ \cite{r-s-pzj} of degree $n-1$ in
the case of even size $N=2n$, and degree $n$ in the odd case
$N=2n+1$. 

Using the projectors $p_i^\pm= \frac{1\pm\sigma_i^z}{2}$ of the $i$-th
spin in the  up/down direction, let us write the exchange
eqs.(\ref{qKZ1}) in components. 
If we have a spin up at
site $i$ and a spin down at site $i+1$ or viceversa, we have
\eq\label{triang-qKZ}
\begin{split}
p_i^- p_{i+1}^+ \Psi^\mu
(z_{i},z_{i+1})&=\sigma_i^- \sigma_{i+1}^+ \frac{(q
  z_{i}-q^{-1}z_{i+1})\Psi^\mu
(z_{i+1},z_{i})
  -(q-q^{-1})z_{i}\Psi^\mu
(z_{i},z_{i+1})}{z_{i+1}-z_{i}} \\
p_i^+ p_{i+1}^-\Psi^\mu
(z_{i},z_{i+1})&=\sigma_i^+ \sigma_{i+1}^-\frac{(q
  z_{i}-q^{-1}z_{i+1})\Psi^\mu
(z_{i+1},z_{i})
  -(q-q^{-1})z_{i+1}\Psi^\mu
(z_{i},z_{i+1})}{z_{i+1}-z_{i}}
\end{split}
\en
These equations form a triangular system. Starting from a given
component we can reconstruct all the others repeatedly using
eqs.(\ref{triang-qKZ}). Therefore if we want to show that two a priori
distinct solutions of the qKZ equations actually coincide, it is
enough to check the equality of one of their components.  

When there are two consecutive spins pointing in the same direction
at positions $i$ and $i+1$, then the first of the qKZ equations reads
\eq
\begin{split}
p_i^\pm p_{i+1}^\pm \Psi^\mu
(\dots,z_{i+1},z_{i},\dots) =
\frac{qz_{i+1}-q^{-1}z_i}{qz_{i}-q^{-1}z_{i+1}}p_i^\pm p_{i+1}^\pm
\Psi^\mu
(\dots,z_{i},z_{i+1},\dots)
\end{split}
\en
which means that the components having two consecutive spins up or
down at positions $i$ and $i+1$,  
$p_i^\pm p_{i+1}^\pm \Psi^\mu
(\dots,z_{i+1},z_{i},\dots)$ have a factor
$qz_{i}-q^{-1}z_{i+1}$ 
\eq
p_i^\pm p_{i+1}^\pm\Psi^\mu
(\dots,z_{i},z_{i+1},\dots) =
(qz_{i}-q^{-1}z_{i+1}) \tilde\Psi^{\mu,\pm}_{i, i+1}(\dots,z_{i},z_{i+1},\dots)
\en
and the vectors $\tilde\Psi^{\mu,\pm}_{i, i+1}(\dots,z_{i},z_{i+1},\dots)$ are
symmetric under exchange $z_i\leftrightarrow z_{i+1}$.
Another useful relation is obtained by considering the matrix $e_i$,  
which is proportional to a projector and is a generator of the
Temperley-Lieb algebra
\eq\label{TL-gen}
e_{i}=\left(\begin{array}{cccc}  
0 & 0 & 0 & 0 \\
0 & -q & 1 & 0 \\
0 & 1 & -q^{-1} & 0 \\
0 & 0 & 0 & 0 
\end{array}
\right),~~~~~~~~~~~~\begin{array}{l}  
e_i^2= \tau e_i,~~~~\tau=-q-q^{-1}\\[5pt]
e_ie_{i\pm 1}e_i = e_i\\[5pt]
e_ie_j=e_je_i
~~~\textrm{for}~~~~|i-j|>1
\end{array}
\en
The matrix $e_i$ is preserved under multiplication
by a $\check R$-matrix for any value of the spectral parameter 
\eq
e_i\check R_{i,i+1}(z)=\check R_{i,i+1}(z)e_i =e_i.
\en 
By applying $e_i$ to the left of the first of the qKZ
eqs.(\ref{qKZ1}) we find  
\eq\label{e_i-symm}
e_i \Psi^\mu(\dots,z_{i},z_{i+1},\dots) = e_i
\Psi^\mu(\dots,z_{i+1},z_{i},\dots). 
\en
The components with most consecutive aligned spins have a completely
factorized form 
\eq\label{factor-comp}
\begin{split}
\Psi^e_{\underbrace{\uparrow,\dots,\uparrow}_n,\underbrace{\downarrow,\dots,\downarrow}_n}({\bf z})&= \prod_{1\leq i< j\leq
  n}\frac{qz_i-q^{-1}z_{j}}{q-q^{-1}} \prod_{n+1\leq i< j\leq
  2n}\frac{qz_i-q^{-1}z_{j}}{q-q^{-1}} \\
\Psi^+_{\underbrace{\uparrow,\dots,\uparrow}_{n+1},\underbrace{\downarrow,\dots,\downarrow}_n}({\bf z})&= \prod_{1\leq i< j\leq
  n+1}\frac{qz_i-q^{-1}z_{j}}{q-q^{-1}} \prod_{n+2\leq i< j\leq
  2n+1}\frac{qz_i-q^{-1}z_{j}}{q-q^{-1}}\prod_{i=n+2}^{2n+1}z_i \\
\Psi^-_{\underbrace{\uparrow,\dots,\uparrow}_n,\underbrace{\downarrow,\dots,\downarrow}_{n+1}}({\bf z})&= \prod_{1\leq i< j\leq
  n}\frac{qz_i-q^{-1}z_{j}}{q-q^{-1}} \prod_{n+1\leq i< j\leq
  2n+1}\frac{qz_i-q^{-1}z_{j}}{q-q^{-1}} 
\end{split}
\en
where the residual normalization ambiguity has been fixed by requiring
these components to be equal to $1$ for $z_i=1$.
From eqs.(\ref{factor-comp}) we see that the maximally factorized
components satisfy (among others) the following relations
\eq\label{diff-size-eq}
\begin{split}
  \Psi^+_{\underbrace{\uparrow,\dots,\uparrow}_{n+1},\underbrace{\downarrow,\dots,\downarrow}_n}({\bf
    z})&= (1-q^{-2})^{n}\Psi^e_{\underbrace{\uparrow,\dots,\uparrow}_{n+1},\underbrace{\downarrow,\dots,\downarrow}_{n+1}}({\bf 
  z})|_{z_{2n+2}=0}\\[5pt]
\Psi^e_{\underbrace{\uparrow,\dots,\uparrow}_n,\underbrace{\downarrow,\dots,\downarrow}_n}({\bf z})&= (1-q^{2})^n \lim_{z_{2n+1}\rightarrow
  \infty}z_{2n+1}^{-n}\Psi^-_{\underbrace{\uparrow,\dots,\uparrow}_{n},\underbrace{\downarrow,\dots,\downarrow}_{n+1}}({\bf z}) 
\end{split}.
\en
Using the triangularity of the eqs.(\ref{triang-qKZ}) we can conclude
that the eqs.(\ref{diff-size-eq}) induce equalities 
between components of $\Psi^+_{2n+1}({\bf z})$  or $\Psi^e_{2n}({\bf z})$
and components of $\Psi^e_{2n+2}({\bf z})$ or $\Psi^-_{2n+1}({\bf z})
$ with the last spin down, i.e.  
\eq\label{zero-sp}
\begin{split}
\Psi^+_{2n+1}({\bf z})\otimes e_{2n+2}^\downarrow &=
(1-q^{-2})^np^-_{2n+2}\Psi^e_{2n+2}({\bf z})|_{z_{2n+2}=0}\\ 
\Psi^e_{2n}({\bf z}) \otimes e_{2n+1}^\downarrow &= (1-q^{2})^n
\lim_{z_{2n+1}\rightarrow
  \infty}z_{2n+1}^{-n}p^-_{2n+1}\Psi^-_{2n+1}({\bf z})  
\end{split}
\en
We will use these equations in order to find a relation among the
inhomogeneous versions of the EFP that we shall introduce in Section
\ref{inhom-section}.  

\subsection{Recursion relation}\label{rec-subsect}

We claim that, upon specialization $z_{i+1}= q^2z_i$ the solution of
the qKZ equation for $N$ spins reduces to the solution of the same
system of equations for $N-2$ spins. In order to make the previous
statement more precise we need to introduce some notation. Let $v_i$
be the vectors which are in the image of the projectors proportional
to the generator of the Temperley-Lieb algebra  $e_i$ 
\eq
v_i = e_i^{\uparrow}\otimes e_{i+1}^{\downarrow}-q^{-1}
e_{i}^{\downarrow}\otimes e_{i+1}^{\uparrow},~~~~~~e_iv_i= -(q+q^{-1})v_i
\en    
Introduce the injective map 
\eq
\Phi_N^{(i)} :(\bC^{2})^{\otimes N} \rightarrow (\bC^{2})^{\otimes
  N+2},
\en
which inserts the vector $v_i$ at position $(i,i+1)$ and shift by two
steps the indices of the sites with $j\geq i$, i.e. on a basis 
\eq
\Phi_N^{(i)} (e_1^{a_1}\otimes \cdots \otimes e_{i-1}^{a_{i-1}}
\otimes e_{i}^{a_{i}}\otimes \cdots \otimes e_{N}^{a_{N}}  ) =
e_1^{a_1}\otimes \cdots \otimes e_{i-1}^{a_{i-1}} \otimes v_i \otimes
e_{i+2}^{a_{i}}\otimes \cdots \otimes e_{N+2}^{a_{N}} .
\en
Then we claim that
\begin{proposition}
The solutions of the exchange equations (\ref{qKZ1}), with the 
 ``boundary conditions'' given by eqs.(\ref{factor-comp}) satisfy the
following  recursion relations 
\eq\label{recursion-psi}
\Psi^{\mu}_N({\bf z})|_{z_{i+1}=q^2 z_i} =
(-q)^{f_N(i)}(q^2z_i)^{\delta(\mu)}\prod_{j=1}^{i-1}\frac{qz_{j}-q^{-1}z_i}{q-q^{-1}} 
\prod_{j=i+1}^{N}\frac{q^{3}z_i-q^{-1}z_{j}}{q-q^{-1}}   \Phi_{N-2}^{(i)}
(\Psi^{\mu}_{N-2})({\bf z}_{\{\widehat{i,i+1}\}}) 
\en
with $\delta(e)=\delta(-)=0$, $\delta(+)=1$, $f_N(i)=i-\lfloor
\frac{N}{2}\rfloor$ and ${\bf z}_{\{\widehat{i,i+1}\}}= \{z_1,
\dots,\widehat{z_i},\widehat{z_{i+1}},\dots,z_N\}$, i.e. the ordered set ${\bf
  z}$ from which the variables $z_i$ and $z_{i+1}$ are removed. 
\end{proposition} 
\begin{proof}
If $z_{i+1}=q^{2}z_i$ then
$\check R_{i,i+1}(z_{i}/z_{i+1})$ becomes a projector proportional to
a generator of the Temperley-Lieb algebra 
\eq
\check R_{i,i+1}(q^{-2}) =
\tau^{-1}e_{i}%~~~~~~~\tau=-(q+q^{-1})
\en
Therefore, by specializing the qKZ equations to  $z_{i+1}=q^2z_i$,  we deduce
$$
\Psi^\mu_N(z_{i},z_{i+1}=q^2 z_i)=  \check
R_{i,i+1}(q^{-2})\Psi^\mu_N(q^2 z_{i},z_{i}) = \tau^{-1}e_i
\Psi^\mu_N(q^2 z_{i},z_{i}). 
$$
In particular 
$ \Psi^\mu_N(z_{i},z_{i+1}=q^2 z_i)$ lies in the image
of $\Phi_{N-2}^{(i)}$ and, by the injectivity of this map, there is a
unique  $\tilde \Psi^\mu_N(z_{i},{\bf z}_{\{\widehat{i,i+1}\}})$ such that 
\eq
\Psi^\mu_N(z_{i},z_{i+1}=q^2 z_i) = \Phi_{N-2}^{(i)} \tilde
\Psi^\mu_N(z_{i},{\bf z}_{\{\widehat{i,i+1}\}}).
\en
In order to determine the equations satisfied by $\tilde
\Psi^\mu_N(z_{i},{\bf z}_{\{\widehat{i,i+1}\}})$ we make use of the
following relations among $R$-matrices 
\eq
\begin{split}
e_i \hat R_{i-1}(z_{i+2}/z_i) \hat R_{i}(z_{i+2}/(q^2 z_i)) \hat
R_{i+1}(z_{i+2}/z_{i-1}) &\hat R_{i}(q^2z_{i}/z_{i-1}) \hat
R_{i-1}(z_{i}/z_{i}) e_i = \\[5pt]
\frac{(qz_{i+2}-q^{-1}z_i)(qz_{i}-q^{-1}z_{i-1})}{
  (qz_{i-1}-q^{-1}z_i)(qz_{i}-q^{-1}z_{i+2})}e_i 
&\hat R_{i-1,i+2}(z_{i+2},z_i) e_i. 
\end{split}
\en
Applying both sides to $\Psi^\mu_N(z_{i},z_{i+1}=q^2 z_i)$ and using that 
\eq
\hat R_{i-1,i+2}(z_{i+2},z_i) 
\Phi_i = \Phi_i \hat R_{i-1,i}(z_{i+2},z_i)  
\en
we get that $\tilde \Psi^\mu_N(z_{i},{\bf z}_{\{\widehat{i,i+1}\}})$ 
satisfies
\eq
\begin{array}{c}
$$\displaystyle{
\tilde \Psi^\mu_N(z_{i},{\bf z}_{\{\widehat{i,i+1}\}}) =}$$\\[5pt] 
$$\displaystyle{
\frac{(qz_{i+2}-q^{-1}z_i)(qz_{i}-q^{-1}z_{i-1})}{
  (qz_{i-1}-q^{-1}z_i)(qz_{i}-q^{-1}z_{i+2})} \hat
R_{i-1,i}(z_{i+2},z_i) \tilde \Psi^\mu_N(z_{i},{\bf z}_{\{\widehat{i,i+1}\}}). }$$
\end{array}
\en
and the vector
\eq
(-q)^{-f_N(i)}(q^2z_i)^{\delta(\mu)}\prod_{j=1}^{i-1}\frac{q-q^{-1}}{qz_{j}-q^{-1}z_i}
\prod_{j=i+1}^{N}\frac{q-q^{-1}}{q^{3}z_i-q^{-1}z_{j}} \tilde
\Psi^\mu_N(z_{i},{\bf z}_{\{\widehat{i,i+1}\}})   
\en
satisfies all the qKZ equations at size $N-2$. In order to check
that it coincides with $\Psi^\mu_{N-2}({\bf z}_{\{\widehat{i,i+1}\}})$
it is enough to check that the components with most aligned consecutive
spins starting from position $i+1$ coincide, which is indeed the case.  
\end{proof}

\section{(Pseudo)-EFP with spectral parameters}\label{inhom-section}

The formal definition of the Emptiness Formation Probability makes use
of the natural scalar product 
$(\cdot,\cdot)_N$ on $\cH_N$, induced by the scalar product
on $\bC_i^2$ where $\{e_i^\uparrow, e_i^\downarrow \}$ form an
orthonormal basis\footnote{In the following most of the time it will
  be clear from the context which Hilbert space we are considering and
  therefore we will omit the label
  $N$ in the scalar product $(\cdot,\cdot)_N$.}. The (pseudo)-EFPs
read  
\eq\label{hom-EFP-def}
\begin{split}
E^{\pm}_{2n+1}(k)&= \frac{\left(\Psi^\pm_{2n+1}({\bf 1}),\prod_{i=1}^k
  p^+_i \cdot\Psi^\pm_{2n+1}({\bf
    1})\right)}{\left(\Psi^\pm_{2n+1}({\bf 1}), 
  \Psi^\pm_{2n+1}({\bf 1})\right)}\\ 
E^{e}_{2n}(k)&= \frac{\left((\Psi^e_{2n}({\bf 1}))^*,\prod_{i=1}^k
  p^+_i \cdot\Psi^e_{2n}({\bf 1})\right)}{\left((\Psi^e_{2n}({\bf
    1}))^*, \Psi^e_{2n}({\bf 1} )\right)} \\ 
E^{\tilde e}_{2n}(k)&= \frac{\left(\Psi^e_{2n}({\bf 1}),\prod_{i=1}^k
  p^+_i \cdot\Psi^e_{2n}({\bf 1})\right)}{\left(\Psi^e_{2n}({\bf 1}),
  \Psi^e_{2n}({\bf 1})\right)}.  
\end{split}
\en

Our strategy to compute the EFPs is to consider an inhomogeneous
version of these 
quantities which is obtained, roughly speaking, by substituting in
eqs.(\ref{hom-EFP-def}) $\Psi^\mu_N({\bf 1})$ with $\Psi^\mu_N({\bf
  z})$, solution of the qKZ equations. 
We shall see that if the substitution is done in the proper way, then
the inhomogeneous EFPs turn out to be symmetric polynomials in the
spectral parameters and satisfy certain recursion relations
which completely characterize these functions among the polynomials of
the same degree in the variable $z_i$.

When defining the inhomogeneous EFP for $k$ aligned spins up, it is
convenient to extract the factor $\prod_{1\leq
i<j\leq k}(qy_i-q^{-1}y_j)$ from $\prod_{i=1}^k 
  p^+_i \Psi_N({\bf y}_{\{1,\dots ,k\}}; {\bf z}_{\{1,\dots ,N-k\}})$,
  and to introduce the vectors 
$\Psi_N(k;{\bf y}_{\{1,\dots ,k\}}; {\bf z}_{\{1,\dots ,N-k\}}) \in \cH_{N-k} $
\eq
\left(\bigotimes_{i=1}^k e_i^\uparrow \right)\otimes
\Psi_N(k;{\bf y}_{\{1,\dots ,k\}}; {\bf z}_{\{1,\dots ,N-k\}})
= \frac{\prod_{i=1}^k 
  p^+_i \Psi_N(y_1,\dots, y_k, 
z_{1},\dots, z_{N-k})}{\prod_{1\leq
i<j\leq k}(qy_i-q^{-1}y_j)} 
\en
Let us moreover introduce the operator
\eq
\cP_N({\bf z})=\prod_{i=1}^N (z_i p_i^+ +p_i^-).
\en 
that multiplies each component
of the vector by $z_i$ for a spin-up at position $i$. The last
ingredient we need is the $*$ operation which consists in substituting
$q$ with $q^{-1}$. Our definition of the inhomogeneous (and
unnormalized) EFP is the following 
\eq\label{def-inhom-EFP}
\begin{array}{c}
\cE^{\mu}_N(k;{\bf y}_{\{1,\dots ,2k \}};{\bf z}_{\{1,\dots ,N-k \}}) = \\[10pt]
\prod_{i=1}^{N-k}z_i^{-\delta(\mu)}(\cP_{N-k}({\bf
  z})(\Psi^{\mu}_{N}(k;q^{-6} {\bf y}_{\{k+1,\dots ,2k \}};   
{\bf z}))^*
,\Psi^\mu_N(k;{\bf y}_{\{1,\dots ,k\}}; {\bf z}) )_{N-k} 
\end{array}
\en
where $\delta(e)=\delta(-)=0$, while $\delta(+)=1$ and $N$ has the
parity corresponding to $\mu$. For the moment it is evident that for
$\mu=e$ or $\mu=-$, $\cE^{\mu}_N(k;{\bf y}; {\bf z})$ are
polynomials in their variables. Actually the polynomiality is true
also for the case $\mu=+$  as will be shown below. 
We define also the inhomogeneous version of the pseudo-EFP
\eq\label{def-inhom-pseudo}
\begin{array}{c}
\cE^{\tilde \mu}_{N}(k;{\bf y};{\bf z}) = \\[10pt]
\prod_{i=k+1}^{2k} y_i^{\lfloor \frac{N+1}{2}\rfloor -k}\prod_{i=1}^{N-k}
z_i^{\lfloor \frac{N+1}{2}\rfloor-1}  (\Psi^{\mu}_{N}(k;q^{-6}{\bf
  y}^{-1}_{\{k+1,\dots,2k \}} ; {\bf z}^{-1}) ,\Psi^\mu_{N}(k;{\bf
    y}_{\{1,\dots, k\}}; {\bf z}) )_{N-k}.  
\end{array}
\en

The choice  to multiply the variables $y_{k+1},\dots, y_{2k}$ by $q^{-6}
$ is motivated by the fact that in this way $\cE^{\mu}_N(k;{\bf y};
{\bf z})$ turns out to be symmetric under
exchange $y_i \leftrightarrow y_j$ for all $1\leq i,j\leq 2k$, as will
be shown at the end of Section \ref{rec-section-q-gen}. 
The polynomials (\ref{def-inhom-EFP}, \ref{def-inhom-pseudo}) have
other remarkable properties, but for the 
moment let us simply notice that for $q=e^{2\pi i/3}$ and
$z_i=y_i=1$ these functions reduce to the unnormalized version of the
EFPs as defined in eqs.(\ref{hom-EFP-def}).

\vskip .5cm

\noindent
Apparently, looking at
eqs.(\ref{def-inhom-EFP},\ref{def-inhom-pseudo}), it seems that we
have six different families of polynomials under
consideration. Actually in the odd size case we have 
\eq\label{equality-tilde}
\cE^{\mu}_{2n+1}(k;{\bf y}; {\bf z})=\cE^{\tilde \mu}_{2n+1}(k;{\bf
  y};{\bf z}).
\en
This follows from the fact that
\eq\label{rel-odd-Psi}
\begin{split}
\prod_{i=1}^{2n+1} z_i^{n+1}
\Psi^{+}_{2n+1}({\bf z}^{-1})= \cP_{2n+1}({\bf
  z})(\Psi^{+}_{2n+1}({\bf z}))^*\\ 
\prod_{i=1}^{2n+1} z_i^{n}
\Psi^{-}_{2n+1}({\bf z}^{-1})= \cP_{2n+1}({\bf
  z})(\Psi^{-}_{2n+1}({\bf z}))^*. 
\end{split}
\en
To prove eqs.(\ref{rel-odd-Psi}) we observe that the vector
$\Psi^{\mu}_{N}({\bf z}^{-1})$ satisfies the exchange equation
\eq\label{qKZinv}
\Psi^{\mu}_{N}(z_{i+1}^{-1},z_{i}^{-1}) =
\check{R}_{i,i+1}(z_{i}/z_{i+1}) \Psi^{\mu}_{N}(z_{i}^{-1},z_{i+1}^{-1}).
\en
The same equation equation holds also for the
vector $\cP_N({\bf z})(\Psi^{\mu}_{N}({\bf z}))^*$, i.e.
\eq\label{qKZstar}
\cP_N(z_{i+1},z_i)(\Psi^{\mu}_{N}(z_{i+1},z_i))^* =
\check{R}_{i,i+1}(z_{i}/z_{i+1}) \cP_N(z_i,z_{i+1})(\Psi^{\mu}_{N}(z_{i},z_{i+1}))^*. 
\en
This is a consequence of the following commutation relation
among the $\check R$-matrix and the operator $(p_i^+
+z_ip_i^-)(p_{i+1}^+ +z_{i+1}p_{i+1}^-)$ 
$$
\check{R}_{i,i+1}(z_{i}/z_{i+1})(z_{i}p_{i}^+
+p_{i}^-)(z_{i+1}p_{i+1}^+ +p_{i+1}^-) = (z_{i+1}p_i^+
+p_i^-)(z_{i}p_{i+1}^+ +p_{i+1}^-)\check{R}^*_{i,i+1}(z_{i+1}/z_{i}).   
$$
which implies
\eq\label{symm-R}
\check{R}_{i,i+1}(z_{i}/z_{i+1})\cP(z_{i},z_{i+1}) =
\cP(z_{i+1},z_{i})\check{R}^*_{i,i+1}(z_{i+1}/z_{i})    
\en
and eq.(\ref{qKZstar}). Therefore to conclude eqs.(\ref{rel-odd-Psi})
it is sufficient to check that they hold for the components with most aligned
spins.

We are left with only \emph{four} different inhomogeneous EFP and as a bonus we
have also shown that $\cE^{+}_{2n+1}(k;{\bf y}; {\bf z})$ is a
polynomial of its variables. 

\vskip .5cm
\noindent
{\bf Symmetry under $z_i\leftrightarrow z_j$}

\noindent
The inhomogeneous (pseudo)-EFP $\cE^{\mu}_N(k;{\bf y};{\bf z})$
is  obviously symmetric under exchange $y_i \leftrightarrow y_j$  for
$1\leq i,j\leq k$ and $k+1\leq i,j\leq 2k$. Using
eqs.(\ref{qKZinv},\ref{qKZstar}) it is easy to show that it is
symmetric also  under exchange $z_i\leftrightarrow z_j$.  Indeed 
\eq
\begin{array}{c}
\cE^{\mu}_N(k;{\bf y}; \dots,  z_i, z_{i+1},\dots) =\\[5pt] 
((\cP(z_i,z_{i+1})\Psi^{\mu}_{N}(k;z_i, z_{i+1} ))^* , 
\check{R}_{i,i+1}(z_i/z_{i+1})  
\Psi^\mu_N(k; z_{i+1}, z_{i})_N= 
\\[5pt]
(\check{R}_{i,i+1}(z_i/z_{i+1})(\cP(z_i,z_{i+1})\Psi^{\mu}_{N}(k;
z_i, z_{i+1} ))^* ,  \Psi^\mu_N(k;  z_{i+1}, z_{i})_N=   
\\[5pt]
\cP(z_{i+1},z_{i})\Psi^{\mu}_{N}(k; z_{i+1}, z_{i}))^*, 
\Psi^\mu_N(k; z_{i+1}, z_{i})_N=\\[5pt]
\cE^{\mu}_N(k;{\bf y}; \dots, z_{i+1}, z_{i},\dots)
\end{array}
\en
where in the third equality we have used the fact that the $\check
R$-matrix is symmetric while the fourth equality follows from
eq.(\ref{qKZstar}). 
The proof of the symmetry of the pseudo-EFP under $z_i\leftrightarrow
z_j$ is completely analogous. 

\vskip .5cm
\noindent
{\bf Factorized cases}

\noindent
Using eqs.(\ref{factor-comp}) we can provide the value of
$\cE^{\mu/\tilde \mu }_{N}(k;{\bf y}; {\bf z})$  corresponding to the
maximal number of consecutive aligned spins. They coincide for the true
and for the pseudo EFP and read
\eq\label{initial-value}
\begin{split}
\cE^{e/\tilde e}_{2k}(k;{\bf y};{\bf z}) 
&= \prod_{1\leq i<j\leq
  k}\frac{(qz_i-q^{-1}z_j)(qz_j-q^{-1}z_i)}{(q-q^{-1})^2}  \\
\cE^{+/\tilde +}_{2k+1}(k+1;{\bf y}; {\bf z}) &=\prod_{1\leq i<j\leq
  k}\frac{(qz_i-q^{-1}z_j)(qz_j-q^{-1}z_i)}{(q-q^{-1})^2}\prod_{i=1}^kz_i^2\\
\cE^{-/\tilde -}_{2k+1}(k;{\bf y};{\bf z}) &=\prod_{1\leq i<j\leq
  k+1}\frac{(qz_i-q^{-1}z_j)(qz_j-q^{-1}z_i)}{(q-q^{-1})^2}.
\end{split}
\en
We will see in the following that the first and the third of these
equations will provide the starting point of a recursion which 
will be worked out in the next section and which completely
characterize the inhomogeneous EFP.
 
\subsection{Recursion relation for the inhomogeneous
  EFP}\label{rec-section-q-gen} 

We begin this section by presenting some relations among the EFP at
different parities which are obtained by setting one of the spectral
parameters to zero or sending it to infinity
\begin{align}\label{special-0empt}
\cE^{+}_{2n+1}(k;{\bf y}; {\bf z} \setminus z_{2n+2})&=
(-1)^n(q-q^{-1})^{-2n}
\left(\prod_{i=1}^{2n+1-k}z_i^{-1}\right)\cE^{e}_{2n+2}(k;{\bf y};
 {\bf z})|_{z_{2n+2}=0}  \\[7pt]\label{special-0empt2}
\cE^{e}_{2n}(k;{\bf y}; {\bf z} \setminus z_{2n+1})&= 
(-1)^n(q-q^{-1})^{2n} \lim_{z_{2n+1}\rightarrow \infty} z_{2n+1}^{-2n}
\cE^{-}_{2n+1}(k;{\bf y};  {\bf z}).
\end{align}
The first of these equation follows from eqs.(\ref{zero-sp}) and by
noticing that
\eq
\cP({\bf z})|_{z_{2n+2}=0}=\cP({\bf z}\setminus z_{2n+2})p_{2n+2}^- . 
\en
For the second one we notice that, writing 
\eq
\begin{array}{c}
\cE^{-}_{2n+1}(k;{\bf y}; {\bf z}) = \\[5pt]
\left((\Psi^{-}_{2n+1}(k; {\bf z}))^* ,
\cP_{2n}({\bf z} \setminus z_{2n+1})z_{2n+1}p^+_{2n+1}  
\Psi^-_{2n+1}(k;  {\bf z}) \right)_{2n+1}~+\\[5pt]
\left((\Psi^{-}_{2n+1}(k;{\bf z}))^* ,
\cP_{2n}({\bf z} \setminus z_{2n+1})p^-_{2n+1}  
\Psi^-_{2n+1}(k;  {\bf z})\right)_{2n+1}
\end{array}
\en
the first term in the r.h.s. is a polynomial in $z_{2n+1}$ of degree
$2n-1$ while the second is of degree $2n$, therefore in the limit only
the second one survives and we can use again eqs.(\ref{zero-sp}).

\vskip .5cm
\noindent
{\bf Specialization $z_i=q^{\pm 2} z_j$}

\noindent

The inhomogeneous EFP satisfies a recursion relation inherited from
the the recursion relation among solutions of the qKZ equations,
eq.(\ref{recursion-psi})  
\eq\label{comp1-rec}
\begin{array}{c}
$$\displaystyle{
\cE^{\mu}_{N}(k;{\bf y}; {\bf z} )|_{z_{i+1}=q^{2}z_i}=}$$\\[5pt]
$$\displaystyle{\tau^{-1} 
( (\Psi^{\mu}_{N}(k; z_i, z_{i+1}=q^{-2}z_{i} ))^*
,\cP(z_{i+1}=q^2z_{i}) e_i
\Psi^\mu_N(k;  z_{i}, z_{i+1}=q^2z_{i}) =}$$\\[5pt]
$$\displaystyle{\tau^{-1} 
( e_i \cP(z_{i+1}=q^2z_{i})(\Psi^{\mu}_{N}(k; z_i,
z_{i+1}=q^{-2}z_{i} ))^* 
,\Psi^\mu_N(k;  z_{i}, z_{i+1}=q^2z_{i}).}$$
\end{array}
\en
A simple computation shows that 
$e_i (p_i^++ z_i p_i^-)(p_i^++ q^2 z_i p_i^-) =\frac{1}{\tau} e_i
(p_i^++ z_i p_i^-)(p_i^++ q^2 z_i p_i^-)e_i^*  $ which means
\eq
e_i \cP(z_{i+1}=q^2z_{i})=\frac{1}{\tau} e_i
\cP(z_{i+1}=q^2z_{i})e_i^* 
\en
and we can substitute it into the last line of eq.(\ref{comp1-rec})
obtaining 
\eq
\tau^{-2} ( e_i \cP(z_{i+1}=q^2z_{i})(e_i\Psi^{\mu}_{N}(k; z_i, z_{i+1}=q^{-2}z_{i} ))^*
,\Psi^\mu_N(k; z_{i}, z_{i+1}=q^2z_{i})_N
\en
Now use eq.(\ref{e_i-symm}) in order to exchange the variables $z_i$
and $z_{i+1}$ in the l.h.s. of the scalar product
\eq
\tau^{-2} ( e_i \cP(z_{i+1}=q^2z_{i})(e_i\Psi^{\mu}_{N}(k;z_{i+1}=q^{-2}z_{i}, z_i))^* 
,\Psi^\mu_N(k;  z_{i}, z_{i+1}=q^2z_{i})_N.
\en
Therefore we can apply to both sides of the scalar product the
recursion relation (\ref{recursion-psi}) and find
\eq\label{rec1}
\begin{array}{c}
$$\displaystyle{
\frac{\cE^{\mu}_{N}(k;{\bf y}; {\bf
    z})|_{z_{i+1}=q^{2}z_i}}{\cE^{\mu}_{N-2}(k;{\bf y};{\bf
    z}\setminus\{z_{i},  z_{i+1}\})}= 
}$$
\\[15pt]
$$\displaystyle{(-1)^k
(1+q^2)z_i\prod_{j=1}^{2k}\frac{qy_{j}-q^{-1}z_i}{q-q^{-1}}\prod_{\substack{1\leq j\leq 
    N-k\\j\neq
    i,i+1}}\frac{(qz_{j}-q^{-1}z_i)(q^{3}z_i-q^{-1}z_{j})}{-(q-q^{-1})^2}.}$$  
\end{array}
\en
The case of the pseudo EFP at even size is analogous
\eq
\begin{array}{c}
$$\displaystyle{
\cE^{\tilde e}_{2n}(k;{\bf y}; {\bf z})|_{z_{i+1}=q^{2}z_i}=}$$\\[5pt]  
$$\displaystyle{ (\Psi^{e}_{2n}(k; z_i^{-1},
    z_{i+1}^{-1}=q^{-2}z_{i}^{-1} ) 
,e_i \Psi^e_{2n}(k;  z_{i}, z_{i+1}=q^2z_{i})_{2n} =}$$\\[5pt]  
$$\displaystyle{ ( e_i \Psi^{e}_{2n}(k;
    z_{i+1}^{-1}=q^{-2}z_{i}^{-1}, z_i^{-1})  
,\Psi^e_{2n}(k; z_{i}, z_{i+1}=q^2z_{i})_{2n} }$$
\end{array}
\en
and again we can apply the recursion at the level of vectors  to both
sides of the scalar product finding
\eq\label{rec2}
\begin{array}{c}
$$\displaystyle{
\frac{\cE^{\tilde e}_{2n}(k;{\bf y};  {\bf
    z})|_{z_{i+1}=q^{2}z_i}}{\cE^{\tilde e}_{2n-2}(k;{\bf y}; {\bf
    z}\setminus\{z_{i}, z_{i+1}\})}=
}$$\\[15pt]
$$\displaystyle{(-1)^k
(1+q^2)\prod_{j=1}^{2k}\frac{qy_{j}-q^{-1}z_i}{q-q^{-1}}\prod_{\substack{1\leq j\leq 
    N-k\\j\neq
    i,i+1}}\frac{(qz_{j}-q^{-1}z_i)(q^{3}z_i-q^{-1}z_{j})}{-(q-q^{-1})^2}. }$$ 
\end{array}
\en
Let us look at $\cE^{e/\tilde e}_{2n}(k;{\bf y}; {\bf z})$
as polynomials in $z_1$. Their degrees are in both cases less than $2n-1$.
The recursion relations 
eqs.(\ref{rec1},\ref{rec2}) provide the value of
$\cE^{e/\tilde e}_{2n}(k;{\bf y};{\bf z})$ for $2(2n-k-1)$
distinct values of $z_1$ (i.e. for $z_1=q^\pm z_i$). Therefore, for
$n>k$,  by Lagrange interpolation these specializations  
determine uniquely $\cE^{e/\tilde e}_{2n}(k;{\bf y}; {\bf z})$
once $\cE^{e/\tilde e}_{2n-2}(k;{\bf y}; {\bf z})$ is known.

As a first consequence we can argue that $\cE^{e/\tilde e}_{2n}(k;{\bf
  y};  {\bf z})$ is symmetric under exchange $y_i \leftrightarrow
y_j$ for all $1\leq i,j\leq 2k$. Indeed for the case when $n=k$ we
have explicit expressions for $\cE^{e/\tilde e}_{2n=2k}(k;{\bf
  y}; {\bf z})$, given by eqs.(\ref{initial-value}), from which we can
read that they are even independent from ${\bf y}$. The
recursion relations (\ref{rec1},\ref{rec2}) are  symmetric under
exchange $y_i \leftrightarrow y_j$ and therefore by induction, if
$\cE^{e/\tilde e}_{2n-2}(k;{\bf y} ; {\bf z})$ is  symmetric, then
also $\cE^{e/\tilde e}_{2n}(k;{\bf y}; {\bf z})$ is symmetric.

A second important consequence is that any family of polynomials
labeled by $n$ and $k$, which satisfy the following conditions:
\begin{itemize}
\item[-] they are symmetric in the spectral parameters,
\item[-] the degree in each spectral parameter is less than $2n-1$,
\item[-] they coincide with  $\cE^{e/\tilde e}_{2k}(k;{\bf y}; {\bf z})$ for $n=k$,
\item[-] they satisfy the recursion relations
  eqs.(\ref{rec1},\ref{rec2}), 
\end{itemize}
must coincide with $\cE^{e/\tilde e}_{2n}(k;{\bf y}; {\bf z})$. 
This line of reasoning will be adopted in Section
\ref{combinatorial-pol} where we will provide a determinantal
representation of $\cE^{e/\tilde
  e}_{2n}(k;{\bf y}; {\bf z})$ at $q=e^{2\pi i/3}$.

The same arguments holds also for $\cE^{-}_{2n+1}(k;{\bf y}; {\bf
  z})$, because the degree is less than $2n+1$ and we have always 
  enough specialization in order to apply the Lagrange interpolation
  and reconstruct all the $\cE^{-}_{2n+1}(k;{\bf y}; {\bf z})$ starting
  from the initial conditions $\cE^{-}_{2k+1}(k;{\bf y};  {\bf
    z})$. The case of $\cE^{+}_{2n+1}(k;{\bf y}; {\bf z})$ is slightly
  different. Again the degree is bounded by 
  $2n+1$ and this allows to fix $\cE^{+}_{2n+1}(k;{\bf y}; {\bf z})$
  starting from $\cE^{+}_{2k+1}(k;{\bf y};  {\bf z})$, but the problem is that we do not have an
  explicit formula $\cE^{+}_{2k+1}(k;{\bf y}; {\bf z})$, being available only for $n=k-1$.
This apparent problem is bypassed using relations (\ref{special-0empt}). 

\section{Inhomogeneous EFP at $q^{2\pi i/3}$}\label{combinatorial-pol}

In order to introduce the expression of
$\cE^{\mu}_N(k;{\bf y};{\bf z})$ and of $\cE^{\tilde
  \mu}_N(k;{\bf y}; {\bf z})$  
which is best suited for taking the specialization
$z_i=y_\alpha=1$ we analyze first the case $k=0$, in which
there are no variables $y$. 
Let us introduce the Young tableaux
\eq
\lambda(m,r)=\{\lfloor \frac{r}{2} \rfloor,\lfloor \frac{r+1}{2}
\rfloor,\dots,\lfloor \frac{r+i-1}{2} \rfloor,\dots,\lfloor 
\frac{r+m-1}{2} \rfloor \}
\en 
Then we find that the inhomogeneous version of the squared norm or of the
sum of the square of the components is given in terms of the product of
two Schur polynomials
\eq\label{case:k=0}
\begin{array}{c}
\cE^{\mu}_{N}(0;{\bf z}) = 3^{-\lfloor \frac{N}{2}
  \rfloor\left(\lceil \frac{N}{2} \rceil-1 \right)} 
S_{\lambda(N,0)}(z_1,\dots,z_{N})S_{\lambda(N,1)}(z_1,\dots,z_{N})\\[5pt]
\cE^{\tilde e}_{2n}(0;{\bf z}) = 3^{-n(n-1)}
S_{\lambda(2n,0)}(z_1,\dots,z_{2n})^2
\end{array}
\en
The proof of eqs.(\ref{case:k=0}) is quite simple and follows the
pattern discussed at the end of the previous section. 
Eqs.(\ref{case:k=0}) are trivially true for $N=1,2$ (or $n=1$),
moreover their  r.h.s. 
are polynomials in $z_1$ of degree at most $2\lceil \frac{N}{2} \rceil
-1$. 
The Schur polynomials $S_{\lambda(m,r)}(z_1,\dots,z_m)$ satisfy a
recursion relation when one specializes $z_i=q^{\pm}
z_j$ (see for example Appendix B of \cite{biane})
\eq\label{recursion-schur}
S_{\lambda(m,r)}({\bf z})|_{z_i=q^{\pm } z_j}=
(-q^{\mp}z_j)^r\prod_{\substack{\ell=1\\
\ell \neq i,j}}^{m}(z_\ell-q^{\mp}z_j) ~ S_{\lambda(m-2,r))}({\bf z}
\setminus \{ z_i, z_j\}). 
\en
This implies that the r.h.s. of eqs.(\ref{case:k=0}) satisfy the
recursion relations eqs.(\ref{rec1},\ref{rec2}) and therefore
eqs.(\ref{case:k=0}) hold.

\vskip .5cm
\noindent
{\bf Generic value of $k$}

\noindent
The recursion relations eq.(\ref{recursion-schur}) for the Schur
functions $S_{\lambda(m,r)}$ suggests a possible representation also
in the case $k\neq 0$. For the sake of clarity let us focus for a
moment on $\cE^{e}_{2n}(k;{\bf y};{\bf z})$. It is easy to see that any
product of the kind $S_{\lambda(2n,0)}({\bf y}_I,{\bf z})S_{\lambda(2n,1)}({\bf y}_{I^c},{\bf z})$, with $I\subset \{1,\dots,
2k \}$ and $I^c=\{1,\dots,
2k \}\setminus I$, satisfies the recursion relations
(\ref{rec1}), but with a ``wrong'' initial condition. It is
reasonable to hope that an appropriate linear combination of terms
with different choices of $I$ could provide the right initial
condition and  hence $\cE^{e}_{2n}(k;{\bf y};{\bf z})$.     

In order to present how this idea actually works it is convenient to
introduce a bit of notation. 
Let $\tilde\rho,\tilde\sigma $ be strictly increasing infinite
sequences of non negative integers, then consider the following family
of matrices   
\eq
\cM^{(\tilde\rho,\tilde\sigma)}(r,s;{\bf y};{\bf z}) =
\left(
\begin{array}{cccccccc}
z_1^{\tilde\rho_1} & z_1^{\tilde\rho_2} & \dots & z_1^{\tilde\rho_{r+s}}& 0 & 0 & \dots
&0 \\ 
z_2^{\tilde\rho_1} & z_2^{\tilde\rho_2} & \dots & z_2^{\tilde\rho_{r+s}}& 0 & 0 & \dots
&0 \\ 
\vdots & \vdots & \ddots & \vdots & \vdots & \vdots & \ddots &\vdots\\
z_r^{\tilde\rho_1} & z_r^{\tilde\rho_2} & \dots & z_r^{\tilde\rho_{r+s}}& 0 & 0 & \dots
&0 \\ 
0 & 0 & \dots & 0 & z_1^{\tilde\sigma_1} & z_1^{\tilde\sigma_2} & \dots & z_1^{\tilde\sigma_{r+s}} \\ 
0 & 0 & \dots & 0 & z_2^{\tilde\sigma_1} & z_2^{\tilde\sigma_2} & \dots &
z_2^{\tilde\sigma_{r+s}} \\ 
\vdots & \vdots & \ddots & \vdots & \vdots & \vdots & \ddots &\vdots\\
0 & 0 & \dots & 0 & z_r^{\tilde\sigma_1} & z_r^{\tilde\sigma_2} & \dots &
z_r^{\sigma_{r+s}} \\
y_1^{\tilde\rho_1} & y_1^{\tilde\rho_2} & \dots & y_1^{\tilde\rho_{r+s}}& y_1^{\tilde\sigma_1}
& y_1^{\tilde\sigma_2} & \dots & y_1^{\tilde\sigma_{r+s}} \\ 
y_2^{\tilde\rho_1} & y_2^{\tilde\rho_2} & \dots & y_2^{\tilde\rho_{r+s}}& y_2^{\tilde\sigma_1}
& y_2^{\tilde\sigma_2} & \dots & y_2^{\tilde\sigma_{r+s}} \\ 
\vdots & \vdots & \ddots & \vdots & \vdots & \vdots & \ddots
&\vdots\\ 
y_{2s}^{\tilde\rho_1} & y_{2s}^{\tilde\rho_2} & \dots & y_{2s}^{\tilde\rho_{r+s}}& y_{2s}^{\tilde\sigma_1}
& y_{2s}^{\tilde\sigma_2} & \dots & y_{2s}^{\tilde\sigma_{r+s}} 
\end{array}
\right)
\en 
and let us define the following polynomials
\eq
\cS^{(\tilde\rho,\tilde\sigma)}(r,s;{\bf y};{\bf z})= \frac{\det
  \cM^{(\tilde\rho,\tilde\sigma)}(r,s;{\bf y};{\bf z})}{\prod_{1\leq i<j\leq
    r}(z_i-z_j)^2\prod_{1\leq i<j\leq
    2s}(y_i-y_j) \prod_{\substack{1\leq i \leq  r\\ 1\leq j \leq  2s}}(z_i-y_j)}.  
\en
The divisibility of $\det
  \cM^{(\tilde\rho,\tilde\sigma)}(r,s;{\bf y};{\bf z})$ by $\prod_{1\leq i<j\leq
    r}(z_i-z_j)^2 $  and by $\prod_{1\leq i<j\leq
    2s}(y_i-y_j)$ is immediate. If we set $z_i=y_j$ for some $i,j$
  then we subtract from the row corresponding to $y_j$ the two rows
  corresponding to $z_i$ getting a null row. This means that $\det
  \cM^{(\tilde\rho,\tilde\sigma)}(r,s;{\bf y};{\bf z})$ is also divisible
  by  \hbox{$(z_i-y_j)$}.
 
Using the Laplace expansion along the first $r+s$
columns we can write $\cS^{(\tilde\rho,\tilde\sigma)}(r,s;{\bf y};{\bf
  z})$ as a bilinear in Schur polynomials   
\eq\label{laplace-exp}
\cS^{(\tilde\rho,\tilde\sigma)}(r,s;{\bf y};{\bf z}) =
\sum_{\substack{I\subset \{1,\dots 2s \}\\|I|=s}}
(-1)^{\epsilon(I)}\frac{\displaystyle{ \prod_{i,j \in I \&
  i<j}(y_i-y_j)\prod_{i,j \in I^c \& i<j}(y_i-y_j)} }{\displaystyle{
    \prod_{1\leq i<j \leq 2s}(y_i-y_j)}}
S_{\rho(r+s)}({\bf z},{\bf y}_I)S_{\sigma(r+s)}({\bf z},{\bf y}_{I^c}),
\en
where $\rho(m)$ and $\sigma(m)$ are Young tableaux of length $m$, whose
entries are $\rho(m)_i= \tilde \rho_i -i+1$, $\sigma(m)_i= \tilde
\sigma_i -i+1$. 
In particular notice that when $s=0$ then
$\cS^{(\tilde\rho,\tilde\sigma)}(r,0;{\bf z})$ factorizes as product of
two Schur polynomials.
  
Now let us introduce the following family of integer sequences
\eq
\tilde \lambda_i(r)=\lfloor \frac{3i-3+r}{2}\rfloor ,~~~~~~
\begin{array}{c}
\tilde\lambda(0)=\{0,1,3,4,6,7,\dots \}\\
\tilde\lambda(1)=\{0,2,3,5,6,8,\dots \}\\
\tilde\lambda(2)=\{1,2,4,5,7,8,\dots \}\\
\cdots
\end{array}
\en
Then we claim that
\begin{align}\label{spectral1-3} 
\cE^{-}_{2n+1}(k;{\bf y};{\bf z}) &=3^{-n^2+k(k-1)/2}  
\cS^{(\tilde\lambda(0),\tilde\lambda(1))}(2n+1-k,k;{\bf y};{\bf z})\\[5pt]\label{spectral1-2} 
\cE^{+}_{2n+1}(k;{\bf y};{\bf z}) &=3^{-n^2+k(k-1)/2} \left(
\prod_{j=1}^{2n-k+1}z_j^{-1} \right)
\cS^{(\tilde\lambda(1),\tilde\lambda(2))}(2n+1-k,k;{\bf y};{\bf z})\\[5pt]\label{spectral1-1}
  \cE^{e}_{2n}(k;{\bf y};{\bf z}) &=3^{-n(n-1)+k(k-1)/2}
\cS^{(\tilde\lambda(0),\tilde\lambda(1))}(2n-k,k;{\bf y};{\bf z})\\[5pt]\label{spectral1-4} 
\cE^{\tilde e}_{2n}(k;{\bf y};{\bf z})
&=3^{-n(n-1)+k(k-1)/2} \left(\prod_{j=1}^{2n-k}z_j^{-1}\right)
\cS^{(\tilde\lambda(0),\tilde\lambda(2))}(2n-k,k;{\bf y};{\bf z})
\end{align}
These formulas reduce to eqs.(\ref{case:k=0}) for $k=0$. 
Using the relations among inhomogeneous EFP with different parities
eqs.(\ref{special-0empt},\ref{special-0empt2}) and the explicit form
of the matrices $\cM^{(\tilde\lambda(i),\tilde\lambda(j))}$ we see
easily that eq.(\ref{spectral1-2}) follows from
eq.(\ref{spectral1-1}), which in turn follows from
eq.(\ref{spectral1-3}). Therefore it remains to prove only 
eq.(\ref{spectral1-3}) and eq.(\ref{spectral1-4}).

The  r.h.s. of eq.(\ref{spectral1-3}) and eq.(\ref{spectral1-4}) are
polynomials in $z_i$ respectively of degree $2n+1$ and $2n-2$. Moreover
using  eq.(\ref{recursion-schur}) and the
form of $\cS^{(\tilde\lambda(r_1),\tilde\lambda(r_2))}(m,k;{\bf
  y};{\bf z})$ expressed by eq.(\ref{laplace-exp})   
we can easily obtain the recursion relation 
\eq\label{rec-cS}
\begin{array}{c}
$$\displaystyle{
\cS^{(\tilde\lambda(r_1),\tilde\lambda(r_2))}(m,k;{\bf y};{\bf z})|_{z_i=q^{\pm } z_j} = } 
$$\\[5pt]
$$\displaystyle{
(-q^{\mp}z_j)^{r_1+r_2}\prod_{\substack{\ell=1\\
\ell \neq i,j}}^{m}(z_\ell-q^{\mp}z_j)^2
  \prod_{\alpha=1}^{2k}(y_\alpha-q^{\mp
}z_j)\cS^{(\tilde\lambda(r_1),\tilde\lambda(r_2))}(m-2,k;{\bf y};{\bf z} \setminus \{z_i,z_j \}).}$$ 
\end{array} 
\en
In order to conclude, as explained at the end of Section
\ref{rec-section-q-gen}, it remains to show that
eqs.(\ref{spectral1-3},\ref{spectral1-4}) hold for $n=k$, i.e. we
need to prove that
\eq
\begin{array}{c}
\cS^{(\tilde\lambda(1),\tilde\lambda(2))}(k+1,k;{\bf y};{\bf z})=
\prod_{1\leq i<j \leq k+1}(z_i^2+z_iz_j+z_j^2)\\ 
\cS^{(\tilde\lambda(0),\tilde\lambda(2))}(k,k;{\bf y};{\bf z})=
\prod_{1\leq i<j \leq k}(z_i^2+z_iz_j+z_j^2) 
.
\end{array}
\en 
We proceed by factor exhaustion. A preliminary remark is that both
$\cS^{(\tilde\lambda(1),\tilde\lambda(2))}(2n-k+1,k;{\bf y};{\bf z})$
and $\cS^{(\tilde\lambda(0),\tilde\lambda(2))}(2n-k,k;{\bf y};{\bf z})$,
as polynomials in $y_i$ are of degree $n-k$ and in particular they
vanish as soon as $k>n$. Therefore using the recursion relation
(\ref{rec-cS}) we conclude that 
\eq\label{initial-to-prove}
\cS^{(\tilde\lambda(1),\tilde\lambda(2))}(k+1,k;{\bf y};{\bf z})|_{z_i=q^\pm z_j}=
\cS^{(\tilde\lambda(0),\tilde\lambda(2))}(2n-k,k;{\bf y};{\bf z})|_{z_i=q^\pm z_j}=0 .
\en
Since their degree as polynomials in $z_i$ is respectively $k$ and
$k-1$, this means that we have proven eqs.(\ref{initial-to-prove}) up
to a numerical constant. Such a constant will be fixed to be equal to
$1$ in the following section, where we shall compute explicitly the 
specialization of the inhomogeneous EFP for $z_i=t^{i-1}$ and $y_j =
t^{N-k+j-1}$.

\subsection{Homogeneous limit}\label{hom-sect}

In this section we arrive at last to the computation of the
homogeneous (pseudo) EFP using eqs.(\ref{spectral1-3}-\ref{spectral1-4}). We need only
to consider a last intermediate step by setting ${\bf z}= {\bf z}(t)$
and ${\bf y} = t^{N-k} {\bf y}(t)$ with 
$$z(t)_i=t^{i-1} ~~~~~
 \textrm{and} ~~~~~y(t)_j = t^{j-1}$$
The matrices $\cM^{(\tilde\lambda(r),\tilde\lambda(s))}(m,k;{\bf
  z}(t);t^m{\bf y}(t))$ with $r,s$ and $m$ as in
eqs.(\ref{spectral1-3}-\ref{spectral1-4}) have noticeable structure as
columns matrices. 
Let us look at a concrete example %$\cM^{(\tilde\lambda(0),\tilde\lambda(1))}(3,2;{\bf z}(t);t^3{\bf y}(t))$
\eq
\cM^{(\tilde\lambda(0),\tilde\lambda(1))}(3,2;{\bf z}(t);t^3{\bf y}(t)) =
\left(
\begin{array}{cccccccccc}
t^{0\cdot 0} & t^{1\cdot 0} & t^{3\cdot 0} & t^{4\cdot 0} & t^{6\cdot 0} & 0 & 0 & 0 & 0 & 0 \\ 
t^{0\cdot 1} & t^{1\cdot 1} & t^{3\cdot 1} & t^{4\cdot 1} & t^{6\cdot 1} & 0 & 0 & 0 & 0 & 0 \\ 
t^{0\cdot 2} & t^{1\cdot 2} & t^{3\cdot 2} & t^{4\cdot 2} & t^{6\cdot 2} & 0 & 0 & 0 & 0 & 0 \\ 
0 & 0 & 0 & 0 & 0 & t^{0\cdot 0} & t^{2\cdot 0} & t^{3\cdot 0} &
t^{5\cdot 0} & t^{6\cdot 0}\\  
0 & 0 & 0 & 0 & 0 & t^{0\cdot 1} & t^{2\cdot 1} & t^{3\cdot 1} &
t^{5\cdot 1} & t^{6\cdot 1}\\  
0 & 0 & 0 & 0 & 0 & t^{0\cdot 2} & t^{2\cdot 2} & t^{3\cdot 2} & t^{5\cdot 2} & t^{6\cdot 2}\\ 
t^{0\cdot 3} & t^{1\cdot 3} & t^{3\cdot 3} & t^{4\cdot 3} & t^{6\cdot 3} & t^{0\cdot 3} & t^{2\cdot 3} & t^{3\cdot 3} &
t^{5\cdot 3} & t^{6\cdot 3}\\  
t^{0\cdot 4} & t^{1\cdot 4} & t^{3\cdot 4} & t^{4\cdot 4} & t^{6\cdot 4} & t^{0\cdot 4} & t^{2\cdot 4} & t^{3\cdot 4} &
t^{5\cdot 4} & t^{6\cdot 4}\\
 t^{0\cdot 5} & t^{1\cdot 5} & t^{3\cdot 5} & t^{4\cdot 5} & t^{6\cdot 5} & t^{0\cdot 5} & t^{2\cdot 5} & t^{3\cdot 5} &
 t^{5\cdot 5} & t^{6\cdot 5}\\  
t^{0\cdot 6} & t^{1\cdot 6} & t^{3\cdot 6} & t^{4\cdot 6} & t^{6\cdot 6} & t^{0\cdot 6} & t^{2\cdot 6} & t^{3\cdot 6} &
t^{5\cdot 6} & t^{6\cdot 6}  
\end{array}
\right)
\en 
The entries of the $j$-th column (apart for the zeros) are consecutive
powers of some $v_j$, where $v_j$ is itself a power of $t$ which depends on the
column index $j$. In the precedent example  ${\bf v}=
\{1,t,t^3,t^4,t^6,1,t^2,t^3,t^5,t^6\}$. Moreover some $v_j$ appear
twice, once in the first half of the columns and once in the second half
(in the example $1,t^3,t^6$), while the remaining $v_j$ are of the form $a_1 \lambda^i$
in the first half of the columns and $a_2 \lambda^i$ in the
second half (in the example $\lambda =t^3, a_1=t, a_2=t^2$). 

By these considerations we are led to introduce the following families
of $2(\ell+r+s)\times 2(\ell+r+s)$ matrices $\cG^{(\ell,r,s)}({\bf
  v};\lambda,a_1,a_2)$, that are made of $6$ blocks of rectangular
matrices as follows 
\eq
\cG^{(\ell,r,s)}({\bf v};\lambda,a_1,a_2)=\left(
\begin{array}{cccc}
D_{\ell+r;\ell}^{(0)}({\bf v}) &D_{\ell+r;r+s}^{(0)}(a_1\vec \lambda)
& {\bf 0}& {\bf 0}\\[7pt]
{\bf 0} & {\bf 0} & D_{\ell+r;\ell}^{(0)}({\bf v}) &D_{\ell+r;r+s}^{(0)}(a_2
\vec \lambda)\\[7pt]
D_{2s;\ell}^{(\ell+r)}({\bf v}) &D_{2s;r+s}^{(\ell+r)}( a_1 \vec
  \lambda) & D_{2s;\ell}^{(\ell+r)}({\bf v})
  &D_{2s;r+s}^{(\ell+r)}(a_2\vec \lambda) 
\end{array}
\right),
\en
where ${\bf v}=\{v_1,\dots,v_\ell\}$, $a_i \vec
\lambda=\{a_i,a_i\lambda, \dots, a_i\lambda^{r+s-1}\}$ and  the blocks
consist of the following rectangular matrices 
\eq
D_{m;\ell}^{(j)}({\bf v})=\left( 
\begin{array}{cccc}
v_1^j & v_2^j & \dots & v^j_{\ell} \\
v_1^{j+1} & v_2^{j+1} & \dots & v^{j+1}_{\ell} \\
\vdots & \vdots & \ddots & \vdots \\
v_1^{j+m-1} & v_2^{j+m-1} & \dots & v^{j+m-1}_{\ell} 
\end{array}
\right).
\en
Apart for a trivial reordering of the columns we have
%, the matrices whose
%determinants appear in eqs.(\ref{spectral1-3}-\ref{spectral1-4}), for
%${\bf z}= {\bf z}(t)$ and ${\bf y} = t^{N-k}{\bf y}(t)$, are of the form  
\eq\begin{array}{c}
\cM^{(\tilde\lambda(0),\tilde\lambda(1))}(2n-k,k;{\bf z}(t);{\bf y}(t))=
\cG^{(n,n-k,k)}(\{t^{3i-3}\};t^3,t,t^2)\\[5pt]
\cM^{(\tilde\lambda(0),\tilde\lambda(2))}(2n-k,k;{\bf z}(t);{\bf y}(t))=
\cG^{(n,n-k,k)}(\{t^{3i-2}\};t^3,1,t^2)\\[5pt] 
\cM^{(\tilde\lambda(0),\tilde\lambda(1))}(2n+1-k,k;{\bf z}(t);{\bf y}(t))=
\cG^{(n+1,n-k,k)}(\{t^{3i-3}\};t^3,t,t^2)\\[5pt]
\cM^{(\tilde\lambda(1),\tilde\lambda(2))}(2n+1-k,k;{\bf z}(t);{\bf y}(t))=
\cG^{(n,n-k+1,k)}(\{t^{3i-2}\};t^3,1,t) 
\end{array}.
\en
Therefore, by calling 
\eq
\cE^{\mu}_{N}(k;t):=\cE^{\mu}_{N}(k;{\bf y}(t);{\bf z}(t))
\en
we have
\eq\label{t-spectral}
\begin{array}{c}
\cE^{e}_{2n}(k;t)= 3^{-n(n-1)+k(k-1)/2}\frac{\det \left(
  \cG^{(n,n-k,k)}(\{t^{3i-3}\};t^3,t,t^2) \right)}{\prod_{1\leq i< j
    \leq 2n-k}(t^{j-1}-t^{i-1}) 
  \prod_{1\leq i< j \leq 2n+k}(t^{j-1}-t^{i-1})}\\[8pt]
\cE^{\tilde e}_{2n}(k;t)= 3^{-n(n-1)+k(k-1)/2} \frac{\det \left(
  \cG^{(n,n-k,k)}(\{t^{3i-2}\};t^3,1,t^2) \right)}{\prod_{1\leq i< j
    \leq 2n-k}(t^{j-1}-t^{i-1}) 
  \prod_{1\leq i< j \leq 2n+k}(t^{j-1}-t^{i-1})}\\[8pt]
\cE^{-}_{2n+1}(k;t)= 3^{-n^2+k(k-1)/2} \frac{\det \left(
  \cG^{(n+1,n-k,k)}(\{t^{3i-3}\};t^3,t,t^2) \right)}{\prod_{1\leq i< j
    \leq 2n-k+1}(t^{j-1}-t^{i-1}) 
  \prod_{1\leq i< j \leq 2n+k+1}(t^{j-1}-t^{i-1})}\\[8pt]
\cE^{+}_{2n+1}(k;t)= 3^{-n^2+k(k-1)/2} \frac{\det \left(
  \cG^{(n,n-k+1,k)}(\{t^{3i-2)}\};t^3,1,t) \right)}{\prod_{1\leq i< j
    \leq 2n-k+1}(t^{j-1}-t^{i-1}) 
  \prod_{1\leq i< j \leq 2n+k+1}(t^{j-1}-t^{i-1})}
\end{array}
\en
The remarkable fact about the matrices $\cG^{(\ell,r,s)}({\bf
  v};\lambda,a_1,a_2)$ is that their determinants factorize
nicely. For $r\geq 0$ we have  
\eq\label{factor-glob0}
\det \cG^{(\ell,r,s)}({\bf v};\lambda,a_1,a_2)= \prod_{1\leq i,j \leq
  \ell}(v_i-v_j)^2 \prod_{\alpha=1,2}\prod_{\substack{1\leq i\leq \ell\\
1\leq j \leq r+s}}(v_i-\lambda^{j-1}a_\alpha) \det
\cG^{(0,r,s)}(\lambda,a_1,a_2), 
\en
with
\eq\label{fact-det0-0}
\det \cG^{(0,r,s)}(\lambda;a_1;a_2) =
(a_1a_2)^{\binom{r+s}{2}}\prod_{1\leq i,j\leq
  s}(\lambda^{j-1}a_1-\lambda^{i-1}a_2) ~\cD^{(r,s)}(\lambda),
\en
and
\eq\label{Dlambda-0}
\cD^{(r,s)}(\lambda)=
\frac{(-1)^{s(r+s)}\lambda^{s\left(\binom{r}{2}-\binom{r+s}{2}\right) }\prod_{1\leq
  i<j\leq r} (\lambda^{j-1}-\lambda^{i-1})\prod_{1\leq
  i<j\leq r+2s}(\lambda^{j-1}-\lambda^{i-1})}{\prod_{1\leq i,j\leq
  s}(\lambda^{j+s-1}-\lambda^{i-1})}.  
\en
These facts are proved in all detail in Appendix
\ref{fact-det}.

Before proceeding to the computation of the r.h.s. of
eqs.(\ref{t-spectral}), 
%the limit $t\rightarrow 1$ 
we come back for a moment to the argument we interrupted at the end of 
Section \ref{combinatorial-pol}. Eqs.(\ref{t-spectral}) have been
proven up to a constant independent of the difference $n-k$. To show
that the constant is $1$,
%conclude the proof and show that the constant is $1$ 
it is enough to
check the equations for $\cE^{\tilde e}_{2n}(k;t)$ and for
$\cE^{-}_{2n+1}(k;t)$ hold true in the case $n=k$. For the 
l.h.s. we use eqs.(\ref{initial-value}) with $q=e^{2\pi i/3}$ and
$z_i=t^{i-1}$, while for the r.h.s. we use
eqs.(\ref{factor-glob0}-\ref{Dlambda-0}). Rather than directly
comparing the two sides of the equations, it is more convenient to
compute the double ratio. A tedious but straightforward computation
using Proposition \ref{fact-det0-1} shows that for both sides we have
\eq
\begin{array}{c}
\left(\frac{\cE^{-}_{2k+3}(k+1;t)}{\cE^{-}_{2k+1}(k;t)}\right)/\left(
  \frac{\cE^{-}_{2k+1}(k;t)}{\cE^{-}_{2k-1}(k-1;t)}\right) =
3^{-1}t^{2k}\frac{t^{3(k+1)}-1}{t^{k+1}-1}  \\
\left(\frac{\cE^{\tilde e}_{2k+2}(k+1;t)}{\cE^{\tilde e}_{2k}(k;t)}\right)/\left(
  \frac{\cE^{\tilde e}_{2k}(k;t)}{\cE^{\tilde e}_{2k-2}(k-1;t)}\right) =
3^{-1}t^{2(k-1)}\frac{t^{3k}-1}{t^{k}-1}
\end{array}
\en 
which, combined with the direct verification for $n=k=1,2$, gives the
desired result.

\vskip .5cm
\noindent
{\bf Proofs of the conjectures}

\noindent
Taking the limit $t\rightarrow 1$ directly in eqs.(\ref{t-spectral})
is not easy. Instead we consider the ratios
$\frac{\cE^{\mu}_{N}(k-1,t)}{\cE^{\mu}_{N}(k,t)}$ which are easier to
compute, and from them recover
eqs.(\ref{recurE--},\ref{EFP++},\ref{EFPee*},\ref{EFPee}). Let us
explain the 
computation for the case
$\frac{\cE^{e}_{2n}(k-1,t)}{\cE^{e}_{2n}(k,t)}$, the other case being
dealt with in the same manner. Using the first of
eqs.(\ref{t-spectral}) we get
\eq\label{ratio-t}
\frac{\cE^{e}_{2n}(k-1,t)}{\cE^{e}_{2n}(k,t)} = 3^{1-k}
\frac{\prod_{i=1}^{2n+k-1}t^{2n+k-1}-t^{i-1}}{\prod_{i=1}^{2n-k}t^{2n-k}-t^{i-1}} \times\frac{\det \left(
  \cG^{(n,n-k+1,k-1)}(\{t^{3i-3}\};t^3,t,t^2) \right)}{\det \left(
  \cG^{(n,n-k,k)}(\{t^{3i-3}\};t^3,t,t^2) \right)}
\en
Then from eqs.(\ref{factor-glob0}-\ref{Dlambda-0}) we find that for
generic values of ${\bf v}$, $\lambda$ and $a_i$ the ratio $\frac{\det
  \cG^{(\ell,r,s)}({\bf v};\lambda,a_1,a_2)}{\det 
  \cG^{(\ell,r+1,s-1)}({\bf v};\lambda,a_1,a_2)}$ does not depend on
${\bf v}$ and is given by a very simple formula
\eq
\frac{\det \cG^{(\ell,r,s)}({\bf v};\lambda,a_1,a_2)}{\det
  \cG^{(\ell,r+1,s-1)}({\bf v};\lambda,a_1,a_2)}
=\lambda^{(s-1)(3s-2)/2}\prod_{j=-(s-1)}^{s-1}(\lambda^ja_1-a_2)
\frac{\cD^{(r,s)}(\lambda)}{\cD^{(r-1,s+1)}(\lambda)}  
\en
\eq
\frac{\cD^{(r,s)}(\lambda)}{\cD^{(r-1,s+1)}(\lambda)}  =(-1)^{r+s}
\frac{\prod_{i=1}^{r+2s-1}(\lambda^{i}
  -1)\prod_{i=1}^{s-1}(\lambda^i-1)^2}{\prod_{i=1}^r(\lambda^{i}-1)
  \prod_{i=1}^{2s-1}(\lambda^{i}-1)  \prod_{i=1}^{2s-2}(\lambda^{i}-1)}   
\en
At this point we make use of these equations in eq.(\ref{ratio-t}) and
substitute $\lambda=t^3$, $a_1=t$ and $a_2=t^2$. Repeating the same
steps with the proper modifications for the other EFPs we finally
obtain 
\eq\label{t-conj}
\begin{split}
\frac{\cE^{e}_{2n}(k-1;t)}{\cE^{e}_{2n}(k;t)}&=
t^{\alpha_e(n,k)}
\left(\frac{[3]_t}{3}\right)^{k-1}
\frac{[2n+k-1]_{t}![n-k]_{t^{3}}![2k-1]_{t^{3}}![2k-2]_{t^{3}}!}
     {[2n-k]_t![n+k-1]_{t^{3}}![k-1]_{t^{3}}! 
  [3k-2]_{t}!}     
\\[5pt]
\frac{\cE^{\tilde e}_{2n}(k-1;t)}{\cE^{\tilde e}_{2n}(k;t)}&=
t^{\alpha_{\tilde e}(n,k)}
(-q)\left(\frac{[3]_t}{3}\right)^{k-1}  
\frac{[2n+k-1]_{t}![n-k]_{t^{3}}![2k-1]_{t^{3}}![2k-2]_{t^{3}}!
}{[2n-k]_t![n+k-1]_{t^{3}}![k-1]_{t^{3}}![3k-3]_{t}![3k-1]_{t}}      
\\[5pt]
\frac{\cE^{-}_{2n+1}(k-1;t)}{\cE^{-}_{2n+1}(k;t)}&=
t^{\alpha_-(n,k)}\left(\frac{[3]_t}{3}\right)^{k-1}
\frac{[2n+k]_{t}![n-k]_{t^{3}}![2k-1]_{t^{3}}![2k-2]_{t^{3}}!
}{[2n-k+1]_t![n+k-1]_{t^{3}}![k-1]_{t^{3}}![3k-2]_{t}!}
\\[5pt]
\frac{\cE^{+}_{2n+1}(k-1;t)}{\cE^{+}_{2n+1}(k;t)}&=
t^{\alpha_+(n,k)}\left(\frac{[3]_t}{3}\right)^{k-1}
\frac{[2n+k]_{t}![n-k+1]_{t^{3}}![2k-1]_{t^{3}}![2k-2]_{t^{3}}!
}{[2n-k+1]_t![n+k]_{t^{3}}![k-1]_{t^{3}}![3k-2]_{t}!}
\end{split}
\en
where we have introduced the usual $t$-numbers and $t$-factorials
$$
[n]_t!=\prod_{i=1}^n[i]_t ~~~~~\textrm{and}~~~~~ [i]_t=\frac{t^i-1}{t-1}.
$$
The powers of $t$ in the r.h.s. of eqs.(\ref{t-conj}) do not concern
us because we are actually interested in the specialization $t=1$,
which at this point is immediate and reproduces the conjectured
formulas (\ref{recurE--},\ref{EFP++},\ref{EFPee*},\ref{EFPee}).  

\section*{Acknowledgments}

This work has been supported by the CNRS through a ``Chaire
d'excellence''. 

\appendix

\section{A determinant evaluation}\label{fact-det}

In this appendix we evaluate the determinants of a family of matrices
that appeared in the final step of the computation of the EFP in section
\ref{combinatorial-pol}.
The matrices we are interested are labeled by three indices $\ell, r,
s$ and are made of 
blocks of rectangular matrices
\eq
\cG^{(\ell,r,s)}({\bf v};\lambda,a_1,a_2)=\left(
\begin{array}{cccc}
D_{\ell+r;\ell}^{(0)}({\bf v}) &D_{\ell+r;r+s}^{(0)}(a_1\vec \lambda)
& {\bf 0}& {\bf 0}\\[7pt]
{\bf 0} & {\bf 0} & D_{\ell+r;\ell}^{(0)}({\bf v}) &D_{\ell+r;r+s}^{(0)}(a_2
\vec \lambda)\\[7pt]
D_{2s;\ell}^{(\ell+r)}({\bf v}) &D_{2s;r+s}^{(\ell+r)}( a_1 \vec
  \lambda) & D_{2s;\ell}^{(\ell+r)}({\bf v})
  &D_{2s;r+s}^{(\ell+r)}(a_2\vec \lambda) 
\end{array}
\right),
\en
where ${\bf v}=\{v_1,\dots,v_\ell\}$, $a_i \vec \lambda=\{a_i,a_i\lambda, \dots, a_i\lambda^{r+s-1}\}$ and each
block consists of the following rectangular matrices
\eq
D_{m;\ell}^{(j)}({\bf v})=\left( 
\begin{array}{cccc}
v_1^j & v_2^j & \dots & v^j_{\ell} \\
v_1^{j+1} & v_2^{j+1} & \dots & v^{j+1}_{\ell} \\
\vdots & \vdots & \ddots & \vdots \\
v_1^{j+m-1} & v_2^{j+m-1} & \dots & v^{j+m-1}_{\ell} 
\end{array}
\right)
\en
The matrix $\cG^{(\ell,r,s)}({\bf v};\lambda,a_1,a_2)$ has the total size
$2(\ell+r+s)\times 2(\ell+r+s)$. 
Here is an example 
\eq
\cG^{(2,1,1)}({\bf v};\lambda,a_1,a_2)=\left(
\begin{array}{cccccccc}
1 & 1 & 1 & 1 & 0 & 0 & 0 & 0 \\
v_1 & v_2  & a_1 & \lambda a_1 & 0 & 0 & 0 & 0 \\
v^2_1 & v^2_2  & a_1^2 & (\lambda a_1)^2 & 0 & 0 & 0 & 0 \\
0 & 0 & 0 & 0 & 1 & 1 & 1 & 1 \\
0 & 0 & 0 & 0 & v_1 & v_2  & a_2 & \lambda a_2 \\
0 & 0 & 0 & 0 & v^2_1 & v^2_2  & a_2^2 & (\lambda a_2)^2 \\
v^3_1 & v^3_2  & a_1^3 & (\lambda a_1)^3 & v^3_1 & v^3_2  & a_2^3 & (\lambda
a_2)^3 \\ 
v^4_1 & v^4_2  & a_1^4 & (\lambda a_1)^4 & v^4_1 & v^4_2  & a_2^4 & (\lambda
a_2)^4  
\end{array}
\right)
\en 
We are interested in the determinant of $\cG^{(\ell,r,s)}({\bf
  v};\lambda;a_1;a_2)$
or equivalently in the determinant 
of the matrix $\tilde \cG^{(\ell,r,s)}({\bf
  v};\lambda;a_1;a_2)$, obtained from $\cG^{(\ell,r,s)}( {\bf v};\lambda;a_1;a_2)$
through some simple row and columns
manipulations
\eq\label{def-gtilde}
\tilde \cG^{(\ell,r,s)}({\bf v};\lambda,a_1,a_2)=\left(
\begin{array}{cccc}
D_{\ell+r;\ell}^{(0)}({\bf v}) &D_{\ell+r;r+s}^{(0)}(a_1\vec \lambda)
& {\bf 0}& {\bf 0}\\[7pt]
{\bf 0} &D_{\ell+r+2s;r+s}^{(0)}( a_1 \vec
  \lambda) & D_{\ell+r+2s;\ell}^{(0)}({\bf v})
  &D_{\ell+r+2s;r+s}^{(0)}(a_2\vec \lambda) 
\end{array}
\right).
\en
From the defining eq.(\ref{def-gtilde}) we see that the first $\ell$ columns of 
$\tilde \cG^{(\ell,r,s)}({\bf v};\lambda,a_1,a_2)$ have rank $\ell+r$
and therefore $\det \tilde 
\cG^{(\ell,r,s)}({\bf v};\lambda,a_1,a_2)=0$ for $r<0$. For $r\geq 0$, 
$\det \cG^{(\ell,r,s)}({\bf v};\lambda,a_1,a_2)$ factorizes
nicely. This property is the content of the 
following three propositions.
\begin{proposition}
For $r\geq 0$ we have 
\eq
\det \cG^{(\ell,r,s)}({\bf v};\lambda,a_1,a_2)= \prod_{1\leq i,j \leq
  \ell}(v_i-v_j)^2 \prod_{\substack{\alpha=1,2\\1\leq i\leq \ell\\
1\leq j \leq r+s}}(v_i-\lambda^{j-1}a_\alpha) \det
\cG^{(0,r,s)}(\lambda,a_1,a_2) 
\en
\end{proposition}
\begin{proof}
The proof is done by induction on $\ell$. Let us look at $\det
\cG^{(\ell,r,s)}({\bf v};\lambda;a_1;a_2)$ as a polynomial in
$v_1$. Using  the matrix
$\tilde \cG^{(\ell,r,s)}({\bf v};\lambda;a_1;a_2)$ to compute the
determinant, we see that its degree is at most $2(\ell+r+s-1)$.  
The presence of all the factors $(v_1-v_j)^2$, $(v_1-\lambda^ja_1)$ and
$(v_1-\lambda^ja_2)$ is obvious using the matrix
$\cG^{(\ell,r,s)}({\bf v};\lambda,a_1,a_2)$. Since these factors exhaust the total degree it
only remain to determine the term $\cD^{(\ell,r,s)}({\bf v}\setminus
v_1;\lambda;a_1;a_2)$ constant in $v_1$ 
\eq\label{detZ-fact}
\det
\tilde \cG^{(\ell,r,s)}({\bf v};\lambda;a_1;a_2)=
\prod_{j=0}^{r+s-1}(v_1-\lambda^ja_1)(v_1-\lambda^ja_2)
\prod_{i=2}^\ell(v_1-v_i)^2  
\cD^{(\ell,r,s)}({\bf v}\setminus v_1;\lambda;a_1;a_2).
\en
We can just set $v_1=0$ in $\tilde
\cG^{(\ell,r,s)}({\bf v};\lambda;a_1;a_2)$ and compute its determinant. We get
immediately 
\eq\label{detZ-spe}
\det
\tilde \cG^{(\ell,r,s)}({\bf v};\lambda;a_1;a_2)|_{v_1=0}=
\lambda^{2\binom{r+s}{2}} a_1^{r+s}a_2^{r+s}  \prod_{i=2}^\ell
v_i^2  \det \tilde \cG^{(\ell-1,r,s)}({\bf v}\setminus v_1;\lambda;a_1;a_2).
\en
Comparing eq.(\ref{detZ-fact}) for $v_1=0$ and (\ref{detZ-spe})
we get
\eq
\cD^{(\ell,r,s)}({\bf v}\setminus v_1;\lambda;a_1;a_2)=\det \tilde
\cG^{(\ell-1,r,s)}({\bf v}\setminus v_1;\lambda;a_1;a_2) 
\en
which provides the recursion we were looking for.
\end{proof}
It remains only to establish the evaluation of $\det
\cG^{(0,r,s)}(\lambda;a_1;a_2)$, but let us deal first with the particular
case $r=0$, which is quite simple and was useful in Section
\ref{hom-sect}. When $r=0$ the matrix $\cG^{(0,0,s)}(\lambda;a_1;a_2)$ has the
form of a Vandermonde matrix and we find easily the following 
\begin{proposition}
\eq\label{fact-det0-1}
\det \cG^{(0,0,s)}(\lambda;a_1;a_2) = (a_1a_2)^{\binom{s}{2}}
\prod_{1\leq i<j \leq s}(\lambda^{i-1}-\lambda^{j-1})^2
\prod_{1\leq i,j\leq s}(\lambda^{i-1}a_1-\lambda^{j-1}a_2)
\en
\end{proposition}
The case when $r$ is generic is a bit more subtle and is dealt with in 
the following
\begin{proposition}
\eq\label{fact-det0}
\det \cG^{(0,r,s)}(\lambda;a_1;a_2) =
(a_1a_2)^{\binom{r+s}{2}}\prod_{1\leq i,j\leq
  s}(\lambda^{j-1}a_1-\lambda^{i-1}a_2) ~\cD^{(r,s)}(\lambda)
\en
with
\eq\label{Dlambda}
\cD^{(r,s)}(\lambda)=
\frac{(-1)^{s(r+s)}\lambda^{s\left(\binom{r}{2}-\binom{r+s}{2}\right) }\prod_{1\leq
  i<j\leq r} (\lambda^{j-1}-\lambda^{i-1})
\prod_{1\leq
  i<j\leq r+2s}(\lambda^{j-1}-\lambda^{i-1})}{\prod_{1\leq i,j\leq
  s}(\lambda^{j+s-1}-\lambda^{i-1})}  
\en
\end{proposition} 
\begin{proof}
The proof is done by factor identification.
For $\ell=0$ the matrix $\tilde \cG^{(0,r,s)}(\lambda;a_1;a_2)$ reads as follows
\eq
\left(
\begin{array}{cc}
D_{r;r+s}^{(0)}(a_1\vec \lambda)
& {\bf 0}\\[7pt]
D_{r+2s;r+s}^{(0)}( a_1 \vec
  \lambda)   &D_{r+2s;r+s}^{(0)}(a_2\vec \lambda) 
\end{array}
\right).
\en
From the Laplace expansion of $\det\tilde \cG^{(0,r,s)}(\lambda;a_1;a_2)$ with respect
to the first $r+s$ columns we can easily deduce that all the
determinants of the minors coming from the first $r+s$ columns are
divided by $a_1^{\binom{r+s}{2}}$, while all the determinants of the
minors coming from the last $r+s$ columns are divided by
$a_2^{\binom{r+s}{2}}$. In this way we have obtained the factor
$(a_1a_2)^{\binom{r+s}{2}}$ 
in eq.(\ref{fact-det0}). In order to determine the total degree of
$\det \tilde \cG^{(0,r,s)}(\lambda;a_1;a_2)$ as a function of $a_1$ and $a_2$ we notice
that the only non vanishing contributions in the Laplace expansion of 
$\det \tilde \cG^{(0,r,s)}(\lambda;a_1;a_2)$ with respect to the first $r+s$
columns come from the minor in the first $r+s$ columns containing the first $r$
rows and $s$ among the last $2s$ rows, while the minor in the last
$r+s$ columns containing the rows from $r+1$ to $2r$ 
and $s$ among the last $2s$ rows. Therefore the term of highest degree in
$a_1$ corresponds to the minor containing the last $s$ rows, and its
degree is $\binom{r+s}{2}+s^2$. 

Now we show that as a function of $a_1$, $\det \tilde
\cG^{(0,r,s)}(\lambda;a_1;a_2)$ 
has zeros of order $s-|j|$ at $a_1=\lambda^j a_2$ for $0\leq |j|\leq s-1 $. 
Let for simplicity $j\geq 0$ (the case $j<0$ being dealt with in the
same manner). 
If we set $a_1=\lambda^j a_2 $ then the first $r+s-j$ columns of $D_{r+2s;r+s}^{(0)}( \lambda^j a_2 \vec
  \lambda)$  are equal to the last $r+s-j$ columns of
  $D_{r+2s;r+s}^{(0)}(a_2 \vec \lambda)$. Therefore if in $\tilde
\cG^{(0,r,s)}(\lambda;a_1=\lambda^j a_2;a_2)$ we subtract 
 the $r+s+j+i$-th column from the
$i$-th column, for $1
\leq i \leq r+s-j$, we obtain a matrix whose first $r+s-j$ columns have rank
$r$. This means that the determinant of the original matrix has a
zero of order at least $s-j$ at $a_1=\lambda^j a_2 $. 

Since we have determined all the zeros of $\det \tilde \cG^{(0,r,s)}(\lambda;a_1;a_2)$
as a function of $a_1$ and $a_2$ we have established eq.(\ref{fact-det0})
up to the unknown factor  $\cD^{(r,s)}(\lambda)$ which does not depend
on $a_1$ and $a_2$. In order to determine such a factor, we specialize
$a_1=\lambda^s a_2$. We find that the first $r$ columns of $D_{r+2s;r+s}^{(0)}( \lambda^s a_2 \vec
  \lambda)$  are equal to the last $r$ columns of
  $D_{r+2s;r+s}^{(0)}(a_2 \vec \lambda)$. Hence by subtracting 
 the $r+2s+i$-th column of the matrix $\tilde
\cG^{(0,r,s)}(\lambda;a_1;a_2)$ from its $i$-th column, for $1 \leq i
\leq r$, we obtain the following matrix 
\eq
\left(
\begin{array}{cc}
D_{r;r}^{(0)}(\lambda^sa_2\vec \lambda)
& {\bf *}\\[7pt]
{\bf 0} & D_{r+2s;r+2s}^{(0)}(a_2
\{\lambda^{r+s},\lambda^{r+s+1},\dots,\lambda^{r+2s-1},1,
\lambda,\dots,\lambda^{r+s-1}\})
\end{array}
\right)
\en
whose determinant is simply the product of the determinants of the two
diagonal blocks $\det D_{r;r}^{(0)}$ and $\det D_{r+2s;r+2s}^{(0)}$
\eq\label{detLambda0}
\begin{array}{c}
\det \cG^{(0,r,s)}(\lambda;a_1=\lambda^s a_2;a_2) =\det
D_{r;r}^{(0)} \det D_{r+2s;r+2s}^{(0)} =\\[5pt]
(-1)^{s(r+s)}\lambda^{s\binom{r}{2}}a_2^{\binom{r}{2}+\binom{r+2s}{2}}\prod_{1\leq
  i<j\leq r} (\lambda^{j-1}-\lambda^{i-1})
\prod_{1\leq
  i<j\leq r+2s} (\lambda^{j-1}-\lambda^{i-1})
\end{array}
\en
By comparing eq.(\ref{detLambda0}) with eq.(\ref{fact-det0}) specialized at
$a_1=\lambda^sa_2$ we obtain eq.(\ref{Dlambda}).
\end{proof}

\begin{corollary}
\eq
\frac{\det \cG^{(\ell,r,s)}({\bf v};\lambda,a_1,a_2)}{\det
  \cG^{(\ell,r+1,s-1)}({\bf v};\lambda,a_1,a_2)}
=\lambda^{(s-1)(3s-2)/2}\prod_{j=-(s-1)}^{s-1}(\lambda^ja_1-a_2)
\frac{\cD^{(r,s)}(\lambda)}{\cD^{(r-1,s+1)}(\lambda)}  
\en
\eq
\frac{\cD^{(r,s)}(\lambda)}{\cD^{(r-1,s+1)}(\lambda)}  =(-1)^{r+s}
\frac{\prod_{i=1}^{r+2s-1}(\lambda^{i}
  -1)\prod_{i=1}^{s-1}(\lambda^i-1)^2}{\prod_{i=1}^r(\lambda^{i}-1)
  \prod_{i=1}^{2s-1}(\lambda^{i}-1)  \prod_{i=1}^{2s-2}(\lambda^{i}-1)}   
\en
\end{corollary}

\section{Plane Partitions/Dimer coverings}\label{plane-part}

A Plane Partition can be seen as a tiling of a regular hexagon of side length $N$
with the following three types of rhombi of unit side length
\begin{center}
\includegraphics[scale=.8]{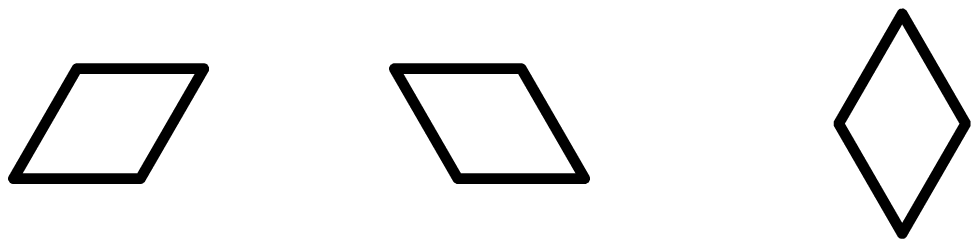}
\end{center}
A $k$-Punctured Cyclically Symmetric Self-Complementary Plane Partition (PCSSCPP) of size $2n$
is a Plane Partition symmetric under a $\pi/3$ rotation around the
center of the hexagon of side length $2n$ and
with a star shaped frozen region of size $k$ (see Figure \ref{figura1}).
 
Following closely Ciucu \cite{ciucu} we compute the enumeration of
$k$-PCSSCPP of size $2n$, that we call $CSSCPP(2n,k)$. It is well
known that Plane Partitions can be 
seen also as Dimer Coverings of an hexagonal graph. In the  case of the
$k$-PCSSCPP the graph is reported in Figure \ref{figura2}. 
\begin{figure}[!h]
\begin{center}
  \includegraphics[scale=.7]{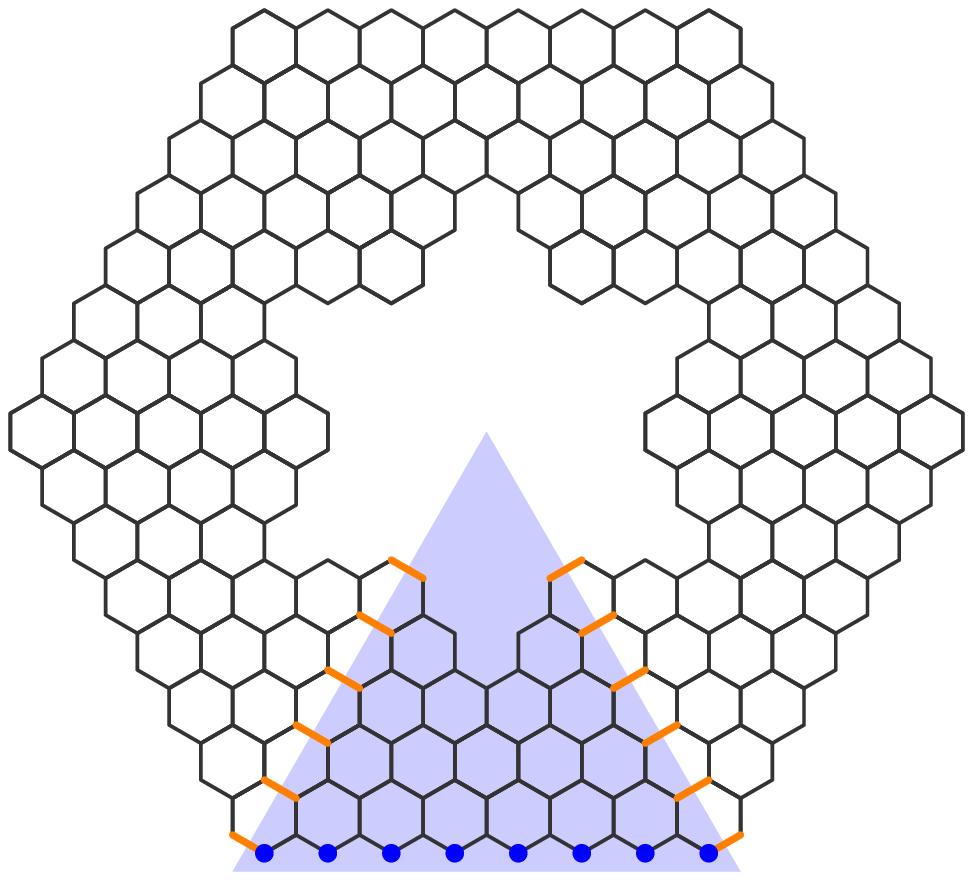}
  ~~~~~~\includegraphics[scale=.5]{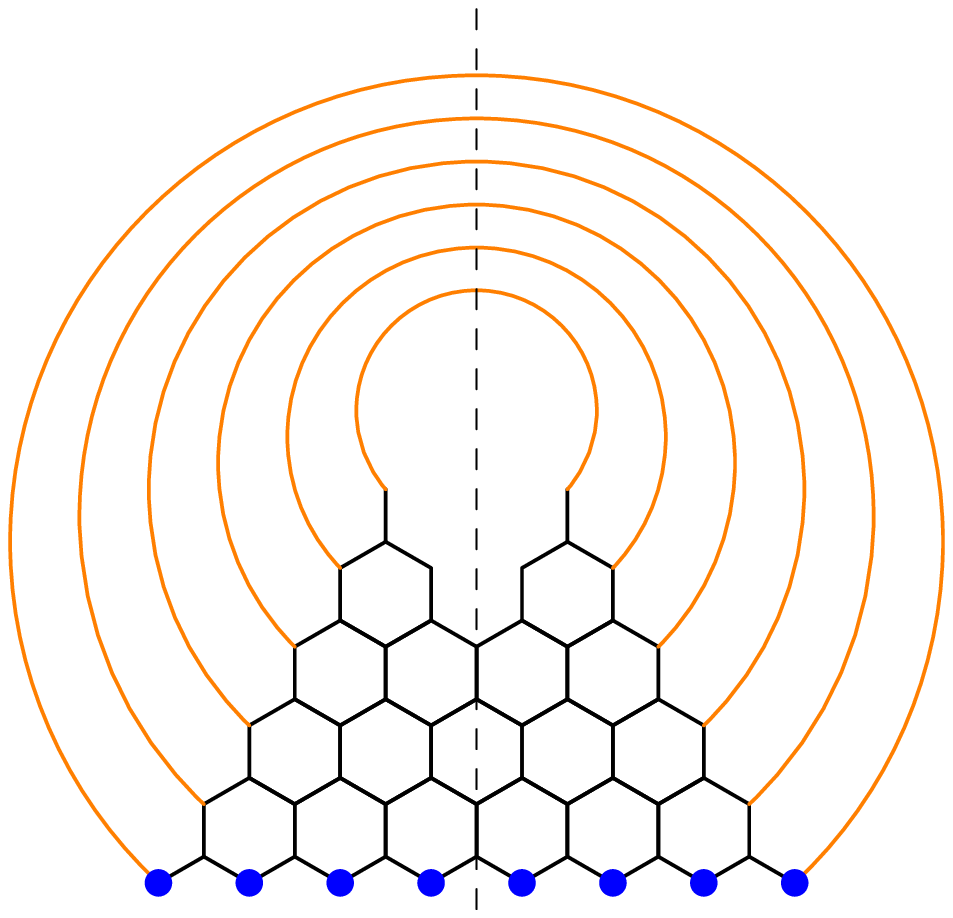}   
\caption[width=.7 \textwidth]{On the left the hexagonal graph corresponding to a
  $PCSSCPP$. The number of orange edges is $2n-k$ , while the number
  of boundary vertices on the bottom side is $2n$.
The shaded region is a fundamental domain. On the right is the graph
$G(n,k)$ obtained by restricting to the fundamental domain and making
the identifications of the orange edges imposed by the cyclic 
symmetry.}\label{figura2}  
\end{center}
\end{figure}
Thanks to the symmetry under a rotation of $\pi/3$,  it is sufficient
to consider dimer coverings  of the fundamental  domain with
``periodic boundary conditions'' which are nothing else than dimer
covering of a graph obtained by cutting the fundamental domain and
joining the opposite cutted edges as on the right of Figure
\ref{figura2}. Let us call $G(n,k)$ this graph. Notice that $G(n,k)$ is a
planar graph symmetric under reflection along the vertical axis. In
\cite{ciucu} Ciucu has proven that $M(G)$, the enumeration of dimer
coverings of a planar graph $G$ with reflection symmetry, is 
related to the weighted enumerations of dimers covering of different
graph $\tilde G$,  obtained from $G$  by
removing the edges incident to the  symmetry axis and lying on its
right (see Figure \ref{figura3}). The relation reads
$$M(G)=2^r M(\tilde G),$$
where $r$ is the number of edges lying on the
symmetry axis, and the weighted enumeration $M(\tilde G)$ is obtained 
by assigning a weight $1/2$ to each dimer lying on the symmetry axis.

In the case of PCSSCPP we are led to
consider weighted dimers covering of the graph $\tilde G(n,k)$
reported in Figure \ref{figura3}, which can also be seen as Rhombus Tilings of
the domain $\tilde D(n,k)$ or as Non-Intersecting Lattice paths
starting from the right boundary of 
$\tilde D(n,k)$ and ending  on its north-west boundary. This last
representation allows to use the Lindstr\"om-Gessel-Viennot theorem
and find 
\eq
CSSCPP(2n,k) = 2^{n-k}M(\tilde G(n,k))= \det [Q_{i+k,j+k}]_{1\leq
  i,j\leq n-k} 
\en
with 
\eq
Q_{i,j}= 2\binom{i+j-2}{2j-i-2}+\binom{i+j-2}{2j-i-1}
\en   
The determinant of $Q_{i+k,j+k}$ can be evaluated using Theorem 40 of
\cite{kratten} and we get
\eq
CSSCPP(2n,k) =  \prod_{h=1}^{n-k} \frac{(j-1)!(j+2k-1)!(3j+3k-2)!(3j+3k-2)!}
    {(2j+k-1)!(2j+k-2)!(2j+3k-1)!(2j+3k-2)!}
\en
from which one easily finds
\eq
\frac{CSSCPP(2n,k-1)}{CSSCPP(2n,k)}=\frac{(2k-2)!(2k-1)!(2n+k-1)!(n-k)!
}{(k-1)!(3k-3)!(3k-1)(2n-k)!(n+k-1)!}.
\en

\vskip 3cm
\begin{figure}[!h]
\begin{center}
  \includegraphics[scale=.5]{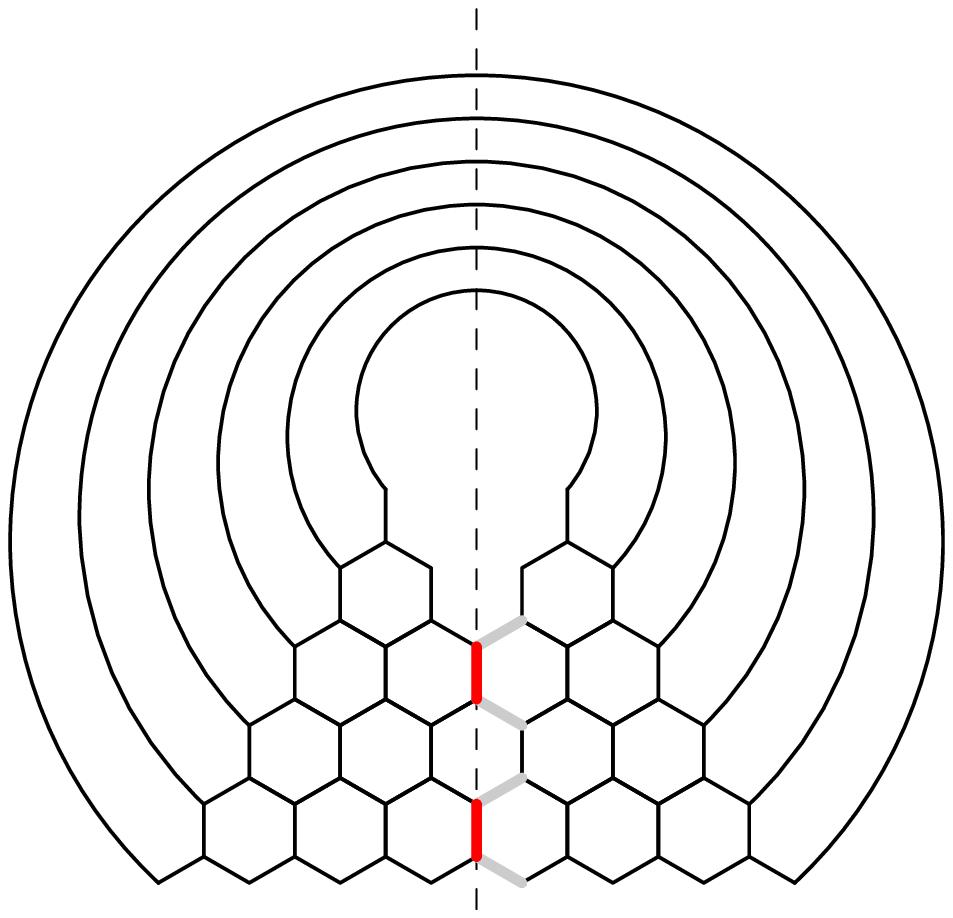} \hskip 3cm
  \includegraphics[scale=.5]{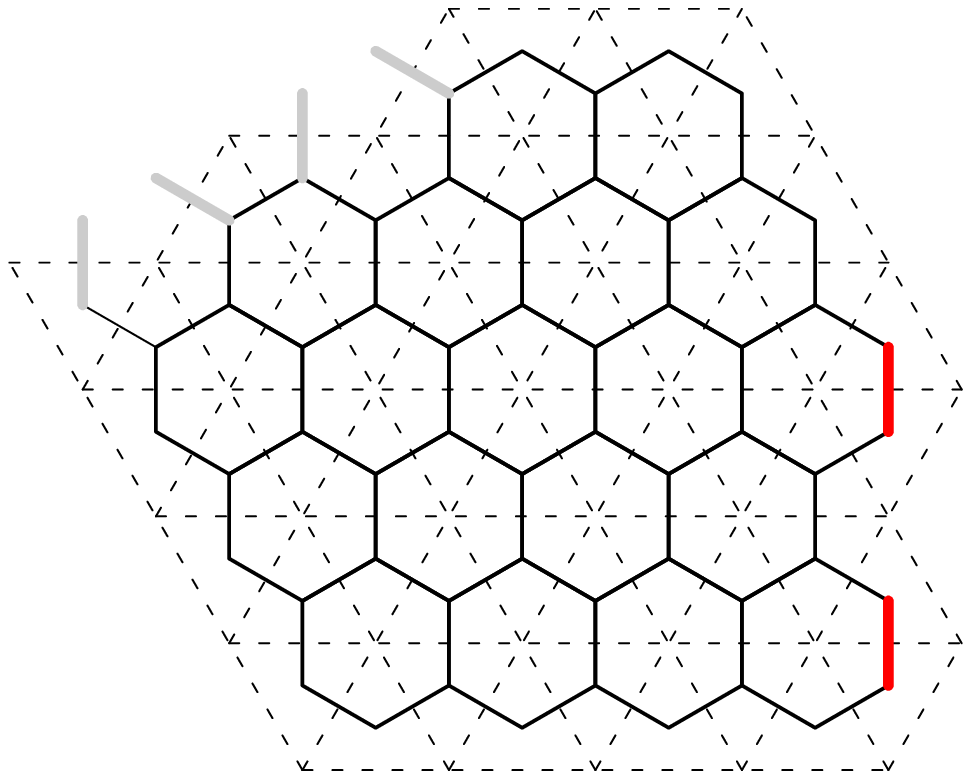}\\ 
\vskip .5cm
\includegraphics[scale=.5]{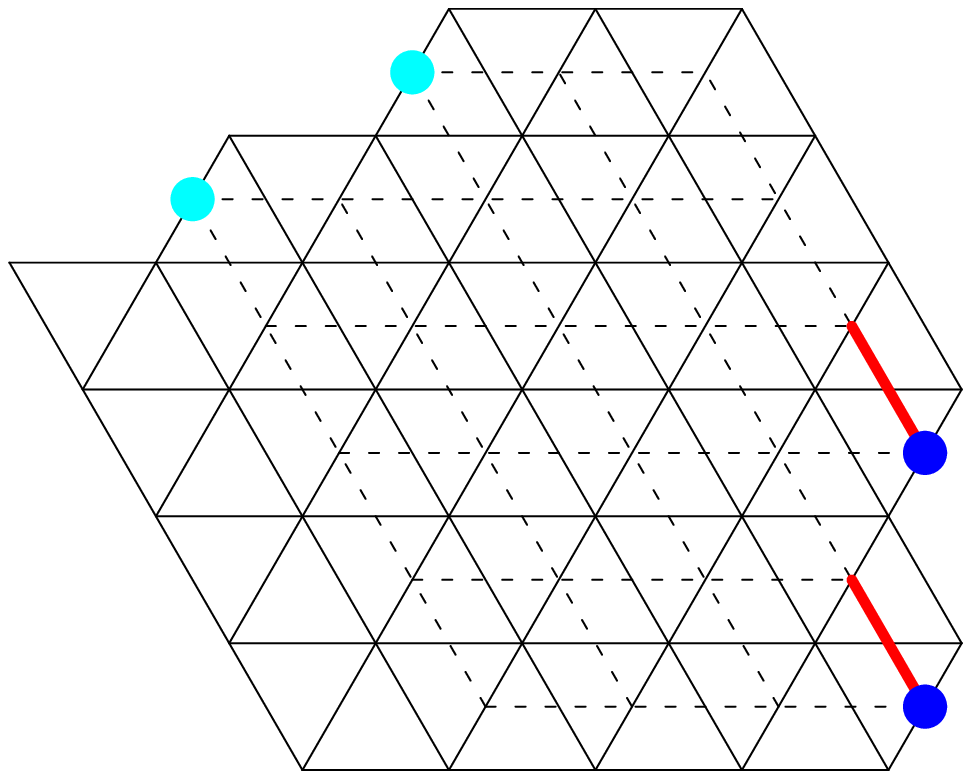} \hskip 3cm
\includegraphics[scale=.5]{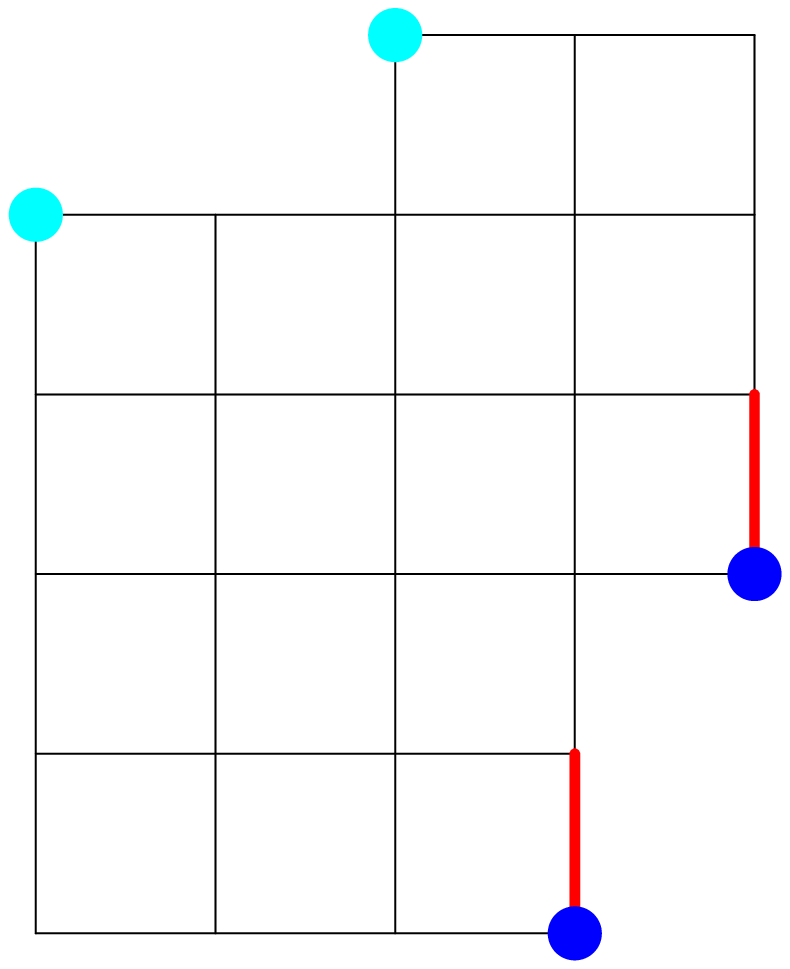} 
\end{center}
\begin{picture}(0,0)
\put(20,300){a)} 
\put(250,300){b)} 
\put(20,130){c)} 
\put(250,130){d)} 
\end{picture}
\caption{a) The domain $\tilde G(n,k)$, obtained from $\tilde G(n,k)$
  by removing the edges incident to the symmetry axis and lying on its
right. The dimers on the red edges, which lye on the symmetry axis,
have weight $1/2$. b) Another presentation of the graph $\tilde
G(n,k)$ and in dashed the corresponding domain $\tilde D(n,k)$. c) The
domain $\tilde D(n,k)$. In dashed is the square lattice on which run
the NILPs starting from the blue points and ending on the cyan
points. 
d) A redrawing of the the square lattice of the NILPs. Each path going
through a red edge gets a weight $1/2$.}\label{figura3}
\end{figure}


\begin{thebibliography}{100}

\bibitem{Baxter} R.J. Baxter, Exactly solved models in statistical mechanics (London: Academic;
New York: Dover).

\bibitem{mccoy} E. Barouch and B. M. McCoy, Phys. Rev. A 3, 786 (1971)

\bibitem{raz-strog1} A. V. Razumov, Yu. G. Stroganov, J.Phys. A34 (2001)
  3185, arXiv:cond-mat/0012141. 

\bibitem{raz-strog2} A. V. Razumov, Yu. G. Stroganov, J.Phys. A34 (2001)
  5335-5340, arXiv:cond-mat/0102247.

\bibitem{BdGN} M. T. Batchelor, J. de Gier and B. Nienhuis, J. Phys. A 34 (2001) L265-L270,
 arXiv.org:cond-mat/0101385.

\bibitem{raz-strogO(1)_1} A. V. Razumov, Yu. G. Stroganov, Theor. Math. Phys. 138
  (2004) 333-337; Teor. Mat. Fiz. 138 (2004) 395-400, arXiv:math.CO/0104216.
  A. V. Razumov, Yu. G. Stroganov, Theor. Math. Phys. 142 (2005)
  237-243; Teor. Mat. Fiz. 142 (2005) 284-292, arXiv:cond-mat/0108103. 

\bibitem{cantini-sportiello} L. Cantini, A. Sportiello,
  J. Combin. Theory Ser. A 118 (2011), 123–146,
 arXiv:math/1003.3376.

\bibitem{pdf-pzj-1} P. Di Francesco, P. Zinn-Justin, E. J. Combi. 12
  (1)(2005), R6, arXiv:math-ph/0410061. 

\bibitem{other-boundaries} P Di Francesco 2005 J. Phys. A:
  Math. Gen. 38 6091, arXiv:math-ph/0504032. P. Zinn-Justin,
  J. Stat. Mech. Theory Exp. 1 (2007), P01007,
  arXiv:math-ph/0610067. 
L. Cantini, arXiv:math-ph/0903.5050. J. de
  Gier, A. Ponsaing, K. Shigechi, J. Stat. Mech. 0904  (2009) P04010,
  arXiv:math-ph/0901.2961. 

\bibitem{ciucu} M. Ciucu, J. Comb. Theory Ser. A 77 (1997), 67–97, 

\bibitem{pdf-pzj-jbz} P. Di Francesco, P. Zinn-Justin, J.-B. Zuber,
  J. Stat. Mech. (2006) P08011, arXiv:math-ph/0603009.  

\bibitem{maillet-emptiness} N. Kitanine, J.M. Maillet, N.A. Slavnov,
  V. Terras, J.Phys. A35 (2002) L385-L391,  arXiv:hep-th/0201134.

\bibitem{other-models} P. Di Francesco, P. Zinn-Justin,
  Comm. Math. Phys. 262 (2) (2006), 459-487,
  arXiv:math-ph/0412031. P. Zinn-Justin, Comm. Math. Phys. 272 (3)
  (2007), 661-682, arXiv:math-ph/0603018. L. Cantini J. Stat. Mech.
  (2007),08, P08012-P08012,  arXiv:math-ph/0703087. 

\bibitem{pasquier} V. Pasquier, Ann. Henri Poincare, vol. 7, (3),
  (2006), 397-421, arXiv:cond-mat/0506075.
 
\bibitem{pdf-pzj-qKZ1} P. Di Francesco, P. Zinn-Justin, J. Phys. A 38
  (2005) L815-L822,  arXiv:math-ph/0508059.

\bibitem{Frenkel-Reshetikhin} I.B. Frenkel and N. Reshetikhin, Commun. Math. Phys. 146
  (1992), 1–60. 

\bibitem{r-s-pzj} A. Razumov, Yu. Stroganov and P. Zinn-Justin,
  J. Phys. A 40 (39) (2007), 11827-11847, arXiv:0704.3542. 

\bibitem{pdf-pzj0703} P. Di Francesco and P. Zinn-Justin,
  Theor. Math. Phys. 154 (3) (2008), 331-348, 
  arXiv:math-ph/070315.  

\bibitem{biane} P. Biane, L. Cantini, A. Sportiello, arXiv:1101.3427.

\bibitem{boundary} M.T. Batchelor, J. de Gier, B. Nienhuis, J.Phys.A
  34 (2001) L265-L270, arXiv:cond-mat/0101385.

\bibitem{fused} P. Zinn-Justin, Comm. Math. Phys. 272 (3) (2007), 661,
  arXiv:math-ph/0603018. 

\bibitem{SU(N)} P. Di Francesco, P. Zinn-Justin, J. Phys. A 38 (48)
  (2005), L815-L822. arXiv:math-ph/0508059. 

\bibitem{XYZ}  V. V. Bazhanov, V. V. Mangazeev, J. Phys. A:
  Math. Gen. 38 (2005) L145–L153. V. Bazhanov, V. V. Mangazeev, J.
Phys. A: Math. Gen. 39 (2006) 12235,
arXiv:hep-th/0602122. V. V. Mangazeev, V. V. Bazhanov, J. Phys. A:
Math. Theor. 43 (2010) 085206, arXiv:0912.2163.
A. V. Razumov and Y. G. Stroganov,
  Theor. Math. Phys. 164 (2010) 977–991,
  arXiv:0911.5030. P. Fendley and C. Hagendorf, J. Phys. A:
  Math. Theor. 43 (2010) 402004. 
C. Hagendorf, P. Fendley, arXiv:1109.4090.


\bibitem{kratten} C. Krattenthaler, S\'eminaire Lotharingien
  Combin. 42 (1999) (``The Andrews Festschrift''),  B42q, arXiv:math/9902004. 

\end{thebibliography}
\end{document}